\documentclass[10pt,onecolumn,twoside]{article}

\usepackage{epsfig}
\usepackage{titlesec}
\usepackage{url} 
\usepackage{booktabs} 
\usepackage{amsfonts} 
\usepackage{amssymb} 
\usepackage{nicefrac} 
\usepackage{microtype} 
\usepackage{mathrsfs}
\usepackage{bm} 
\usepackage{cite} 
\usepackage{comment}
\usepackage{diagbox}
\usepackage{graphicx} 
\usepackage{mathtools} 
\usepackage{algpseudocode}
\usepackage{algorithm}
\usepackage{amsthm} 
\usepackage{enumitem} 
\usepackage{multirow} 
\usepackage{tabularx} 
\usepackage{hhline}
\usepackage{arydshln} 
\usepackage{enumitem} 
\usepackage{siunitx}
\usepackage[font=small,skip=5pt]{caption}
\usepackage{parskip}
\usepackage{comment}
\usepackage[dvipsnames]{xcolor}
\usepackage{fullpage}
\usepackage[hidelinks]{hyperref}

\algtext*{EndIf}
\algtext*{EndFor}
\algtext*{EndWhile}

\newcolumntype{L}[1]{>{\raggedright\arraybackslash}p{#1}}
\newcolumntype{C}[1]{>{\centering\arraybackslash}p{#1}}
\newcolumntype{R}[1]{>{\raggedleft\arraybackslash}p{#1}}

\theoremstyle{plain} 

\newtheorem{theorem}{Theorem}

\newtheorem{assumption}{Assumption}

\def\defn{\,\coloneqq\,}
\def\argmin{\mathop{\mathsf{arg\,min}}} 

\def\lim{\mathop{\mathsf{lim}}} 
\def\min{\mathop{\mathsf{min}}} 
\def\max{\mathop{\mathsf{max}}}

\def\prox{\mathsf{prox}}
\def\log{\mathsf{log}}
\def\zer{\mathsf{zer}}

\def\ebm{{\bm{e}}}

\def\xbm{{\bm{x}}}

\def\ybm{{\bm{y}}}
\def\zbm{{\bm{z}}}

\def\thetabm{{\bm{\theta}}}
\def\zerobm{\bm{0}}
\def\Abm{{\bm{A}}}

\def\Dbm{{\bm{D}}}
\def\Fbm{{\bm{F}}}

\def\phibm{{\bm{\phi}}}

\def\xbmast{{\bm{x}^\ast}}

\def\xbmtilde{{\widetilde{\bm{x}}}}

\def\xbfhat{{\widehat{\mathbf{x}}}}
\def\xbfhat{{\widehat{\bm{x}}}}

\def\nablahat{{\widehat{\nabla}}}
\def\phihat{{\widehat{\bm\phi}}}
\def\phibmhat{{\widehat{\phibm}}}

\def\Tsf{{\mathsf{T}}}

\def\Hsf{{\mathsf{H}}}

\def\Hsf{{\mathsf{H}}}

\def\C{\mathbb{C}}
\def\R{\mathbb{R}}
\def\E{\mathbb{E}}

\def\Tcal{{\mathcal{T}}}

\def\Lcal{{\mathcal{L}}}

\definecolor{pink}{HTML}{000000}
\definecolor{lightgreen}{RGB}{229,251,229}
\definecolor{red}{rgb}{0,0,0}

\title{SGD-Net:~Efficient Model-Based Deep Learning\\with Theoretical Guarantees}

\author{Jiaming~Liu%
\thanks{Department of Electrical \& Systems Engineering, Washington University in St.~Louis, St.~Louis, MO 63130.}
\hspace{0.05em},
Yu~Sun%
\thanks{Department of Computer Science \& Engineering, Washington University in St.~Louis, St.~Louis, MO 63130.}
\hspace{0.05em}, 
Weijie~Gan$^{\dagger}$,
\hspace{0.05em}
Xiaojian~Xu$^{\dagger}$,\\
Brendt~Wohlberg%
\thanks{Theoretical Division, Los Alamos National Laboratory, Los Alamos, NM 87545 USA.}
, and Ulugbek~S.~Kamilov$^{\ast, \dagger}$}

\begin{document}
\date{}
\maketitle

\begin{abstract}
Deep unfolding networks have recently gained popularity in the context of solving imaging inverse problems. However, the computational and memory complexity of data-consistency layers within traditional deep unfolding networks scales with the number of measurements, limiting their applicability to large-scale imaging inverse problems. We propose \emph{SGD-Net} as a new methodology for improving the efficiency of deep unfolding through stochastic approximations of the data-consistency layers. Our theoretical analysis shows that SGD-Net can be trained to approximate batch deep unfolding networks to an arbitrary precision. Our numerical results on intensity diffraction tomography and sparse-view computed tomography show that SGD-Net can match the performance of the batch network at a fraction of training and testing complexity.
\end{abstract}

\section{Introduction}
\label{Sec:Intro}
The recovery of an unknown image from a set of noisy measurements is a central problem in computational imaging. The recovery is traditionally formulated as an \emph{inverse problem} that combines a physical-model characterizing the imaging system with a regularizer imposing a prior knowledge on the unknown image. Over the past years, many regularizers have been proposed as imaging priors, including those based on transform-domain sparsity, low-rank penalty, and dictionary learning~\cite{Rudin.etal1992, Figueiredo.Nowak2001, Elad.Aharon2006, Danielyan.etal2012, Hu.etal2012, Lefkimmiatis.etal2013}.

There has been considerable recent interest in using \emph{deep learning (DL)} in the context of imaging inverse problems (see recent reviews~\cite{McCann.etal2017, Lucas.etal2018, Ongie.etal2020}). Instead of explicitly defining a regularizer, the traditional DL approach is based on training a \emph{convolutional neural network (CNN)} architecture, such as U-Net~\cite{Ronneberger.etal2015}, to invert the measurement operator by exploiting the natural redundancies in the imaging data~\cite{DJin.etal2017, Kang.etal2017, Chen.etal2017, Sun.etal2018, han.etal2018}. \emph{Plug-and-play priors (PnP)}~\cite{Venkatakrishnan.etal2013} and \emph{regularization by denoising (RED)}~\cite{Romano.etal2017} are two well-known alternative approaches to the traditional DL that enable the integration of pre-trained CNN denoisers, such as DnCNN~\cite{Zhang.etal2017}, as image priors within iterative algorithms.
When equipped with advanced CNN denoisers, PnP/RED provides excellent performance by exploiting both the implicit prior, characterized by the denoiser, and the measurement model~\cite{Chan.etal2016, Sreehari.etal2016, Kamilov.etal2017, Buzzard.etal2017, Sun.etal2019a, Reehorst.Schniter2019, Ryu.etal2019, Mataev.etal2019, Wu.etal2020, Liu.etal2020}. \emph{Deep unfolding} is a related approach that interprets the iterations of an image recovery algorithm as layers of a CNN and trains it end-to-end in a supervised fashion (see \emph{``Unrolling''} in~\cite{McCann.etal2017} or \emph{``Neural networks and analytical methods''} in~\cite{Lucas.etal2018}). Unlike in PnP/RED, the CNN in deep unfolding is trained jointly with the measurement model, leading to an image prior optimized for a given inverse problem~\cite{zhang2018ista, Yang.etal2016, Hauptmann.etal2018, Adler.etal2018, Aggarwal.etal2019, Hosseini.etal2019, Chun.etal2020, Yaman.etal2020, Aggarwal.etal2020, Kellman.etal2020}. Despite the recent popularity of deep unfolding, the training of such networks can be a significant practical challenge in applications that require processing of a large-number of measurements. Specifically, the data-consistency layers of these recursive neural networks are based on \emph{batch} processing, which means that they process the \emph{entire set of measurements} at each layer. While this type of batch data processing is known to be suboptimal in traditional large-scale optimization~\cite{Bottou.Bousquet2007, Bertsekas2011, Kim.etal2013, Bottou.etal2018}, the issue has never been addressed in the context of designing deep unfolding networks.

We address this issue by proposing \emph{SGD-Net} as the first deep unfolding methodology to adopt \emph{stochastic processing} of measurements within data-consistency layers. This improves the efficiency of SGD-Net compared to its batch counterparts on large datasets during both training and testing. We implement SGD-Net by unfolding the gradient-based RED algorithm and introducing stochastic approximations to its data-consistency layers. The CNN within SGD-Net is trained in an end-to-end fashion to remove artifacts due to the imaging system and stochastic processing of the measurements. We present a theoretical analysis of SGD-Net that establishes that the network can be trained to approximate the corresponding batch deep unfolding network to desired accuracy. We also demonstrate the practical relevance of SGD-Net by reconstructing images in \emph{intensity diffraction tomography (IDT)}~\cite{Ling.etal18, Wu.etal2020} and sparse-view~\emph{computed tomography (CT)}~\cite{Kak.Slaney1988}. Our results corroborate the effectiveness of SGD-Net in achieving comparable imaging quality to batch deep unfolding networks at a fraction of computational complexity. SGD-Net thus addresses an important gap in the current literature on deep unfolding by providing an efficient framework applicable to a wide variety of imaging problems.

\section{BACKGROUND}
\label{Sec:Background}
\subsection{Inverse problems in computational imaging}
\label{subSec:InvImaging}
Computational imaging problems can usually be posed as the reconstruction of an  unknown image $\xbm\in\C^n$ from a set of corrupted measurements $\ybm\in \C^m$. The reconstruction is often formulated as an inverse problem
\begin{equation}
\label{Eq:RegularizedOptimization}
\xbmast = \argmin_{\xbm} f(\xbm) \quad\text{with}\quad f(\xbm) = g(\xbm) + h(\xbm),
\end{equation}
where $g$ is the data-fidelity term that quantifies consistency with $\ybm$  and $h$ is the regularizer. For example, two widely used functions in the context of imaging problems are the least squares and total variation (TV)~\cite{Rudin.etal1992}
\begin{equation}
\label{Eq:ForwardModel}
g(\xbm) = \frac{1}{2}\|\ybm - \Abm\xbm\|_2^2 \quad\text{and}\quad r(\xbm) = \tau \|\Dbm\xbm\|_1,
\end{equation}
where $\Abm$ is the measurement operator and $\Dbm$ is the discrete gradient operator. The data-fidelity term in~\eqref{Eq:ForwardModel} assumes a linear measurement model $\ybm=\Abm\xbm + \ebm$, where the measurement operator $\Abm \in \C^{m \times n}$ characterizes the response of the imaging system and $\ebm \in \C^m$ is the noise, which
is often assumed to be independent and identically distributed (i.i.d.) Gaussian.

Many popular regularizers, such as the ones based on the $\ell_1$-norm, are nondifferentiable. Proximal algorithms~\cite{Parikh.Boyd2014}, such as the \emph{proximal gradient method (PGM)}~\cite{Eckstein.Bertsekas1992, Figueiredo.Nowak2003, Bect.etal2004, Daubechies.etal2004, Beck.Teboulle2009b} and the \emph{alternating direction method of multipliers (ADMM)}~\cite{Afonso.etal2010, Ng.etal2010, Boyd.etal2011}, enable efficient minimization of nonsmooth functions without differentiating them by using the \emph{proximal operator}
\begin{equation}
\label{Eq:ProximalOperator}
\prox_{\mu h}(\zbm) \defn \argmin_{\xbm}\left\{\frac{1}{2}\|\xbm-\zbm\|_2^2 + \mu h(\xbm)\right\},
\end{equation}
where $\mu > 0$ is a parameter. Note that the proximal operator can be interpreted as the regularized image denoiser for AWGN with noise of variance $\mu$.
\subsection{Image reconstruction using deep learning}
\label{subSec:DeepImaging}
Recently, deep learning has gained popularity due to its effectiveness for solving imaging inverse problems.  A widely used approach first brings the measurements to the image domain and then trains a deep network to map the corresponding low-quality images $\{\xbmtilde_j\}$ to their clean target versions $\{\xbm_j\}$ by solving an optimization problem~\cite{DJin.etal2017, Sun.etal2018, han.etal2018}
\begin{equation}
\label{Eq:cnn_loss}
\argmin_{\thetabm}\frac{1}{M}\sum_{j=1}^{M}\Lcal(\Tcal_{\thetabm}(\xbmtilde_{j}) - \xbm_{j}),
\end{equation}
where $\Tcal_{\thetabm}$ represents the CNN parametrized by $\thetabm$ trained under the loss function $\Lcal$. Popular loss functions include the $\ell_2$-norm and $\ell_1$-norm~\cite{Zhao.etal2017}. For example, in the context of sparse-view CT, prior methods have trained $\Tcal_\thetabm$ for mapping a filtered backprojected image $\xbmtilde$ to a reconstruction $\xbm$ from a fully-sampled groundtruth data~\cite{McCann.etal2017, han.etal2018}.

The idea of end-to-end inversion can be refined by including the measurement operator into the CNN architecture. Inspired by LISTA~\cite{Gregor.LeCun2010}, the corresponding \emph{unfolding algorithms} interpret iterations of a regularized inversion as layers of a CNN and train it end-to-end in a supervised fashion~\cite{Schmidt.Roth2014, Chen.etal2015, Kamilov.Mansour2016, Bostan.etal2018}. In the context of compressive sensing, ADMM-Net~\cite{Yang.etal2016} and ISTA-Net$^{+}$~\cite{zhang2018ista} have considered jointly training the image transforms and shrinkage functions within an unfolded algorithm. A related class of methods~\cite{Schlemper.etal2018, Aggarwal.etal2019, Hosseini.etal2019}, have included a full CNN as a trainable regularizer within an unfolded algorithm. Such unfolding algorithms have been shown to be effective in a number of problems~\cite{Biswas.etal2019, Hosseini.etal2019} and are closely related to the PnP/RED methods (discussed in Section~\ref{subSec:PnPRED}) that also combine the measurement operator and the imaging prior. Their main differences is that the former optimize the parameters in an end-to-end manner, and generally produce higher-quality results with fewer iterations.  However, model-based deep unfolding architectures require the storage of all the measurements, parameters of the measurement operator, and intermediary activation maps at every iteration of the unfolded algorithm. This limits their ability to solve inverse problems where one needs to process high-dimensional data (for example, see relevant discussion in the section \emph{``Memory Requirements''} in~\cite{Schlemper.etal2018}). It is worth mentioning that a recent work~\cite{Kellman.etal2020} has proposed a memory-efficient learning strategy for model-based deep networks by using reverse recalculations. The reverse recalculation strategy is fundamentally different from our usage of stochastic approximations within data-consistency layers. In fact, both approaches are fully complimentary and can be used together for efficient training of model-based deep architectures.

\subsection{Using denoisers as image priors}
\label{subSec:PnPRED}

Since the proximal operator~\eqref{Eq:ProximalOperator} is mathematically equivalent to regularized image denoising, there has been considerable interest in developing denoiser-based iterative algorithms such as PnP~\cite{Venkatakrishnan.etal2013}  and RED~\cite{Romano.etal2017}. The key idea in PnP is to replace the proximal operator with an advanced image denoiser $H_{\sigma}$, where $\sigma > 0$ controls the strength of denoising. This simple replacement enables PnP to regularize the problem by using advanced denoisers, such as BM3D~\cite{Dabov.etal2007} or DnCNN~\cite{Zhang.etal2017}, that do not correspond to any explicit $h$. Recent studies have confirmed the effectiveness of PnP in a range of imaging applications~\cite{Sreehari.etal2016, Brifman.etal2016, Ahmad.etal2019}.

The RED framework is an alternative scheme where the denoiser can sometimes lead to an explicit regularization function~\cite{Romano.etal2017}. In the most general setting, RED algorithms seek a fixed point $\xbmast$ that satisfies
\begin{equation}
\label{Eq:Gx}
G(\xbmast) = \nabla g(\xbmast) + \tau(\xbmast - H_\sigma(\xbmast)) = 0,
\end{equation}
where $\tau > 0$ is the regularization parameter. Equivalently, $\xbmast$ satisfies
\begin{equation}
\label{Eq:Gxset}
\xbmast \in \zer(G) \defn \{\xbm \in \C^n \,|\, G(\xbm) = \zerobm\}.
\end{equation}
When the denoiser is locally homogeneous and has a symmetric Jacobian~\cite{Romano.etal2017, Reehorst.Schniter2019}, the term $\tau(\xbm - H_{\sigma}(\xbm))$ corresponds to the gradient of the RED regularizer $h(\xbm) =  (\tau/2) \xbm^\Tsf(\xbm - H_\sigma(\xbm))$, which enables a simple interpretation of RED as an instance of ~\eqref{Eq:RegularizedOptimization}. The excellent performance of RED together with learned CNN denoisers has been reported in super-resolution, phase retrieval, and compressed sensing~\cite{Metzler.etal2018, Sun.etal2019c}. Additionally, prior work has developed a scalable \emph{online} variant of RED, known as SIMBA, that is well suited for tomographic applications with a large number of projections~\cite{Wu.etal2020}. However, unlike deep unfolding, the CNN in RED is not jointly trained with the measurement model, limiting its ability to capture the non-iid nature of the artifacts within iterations~\cite{Aggarwal.etal2019}. SGD-Net, introduced in the next section, is a natural extension of SIMBA~\cite{Wu.etal2020} towards CNNs trained for artifact-removal, which leads to a scalable and end-to-end trainable deep network.


\begin{figure*}[t]
\begin{center}
\includegraphics[width=16cm]{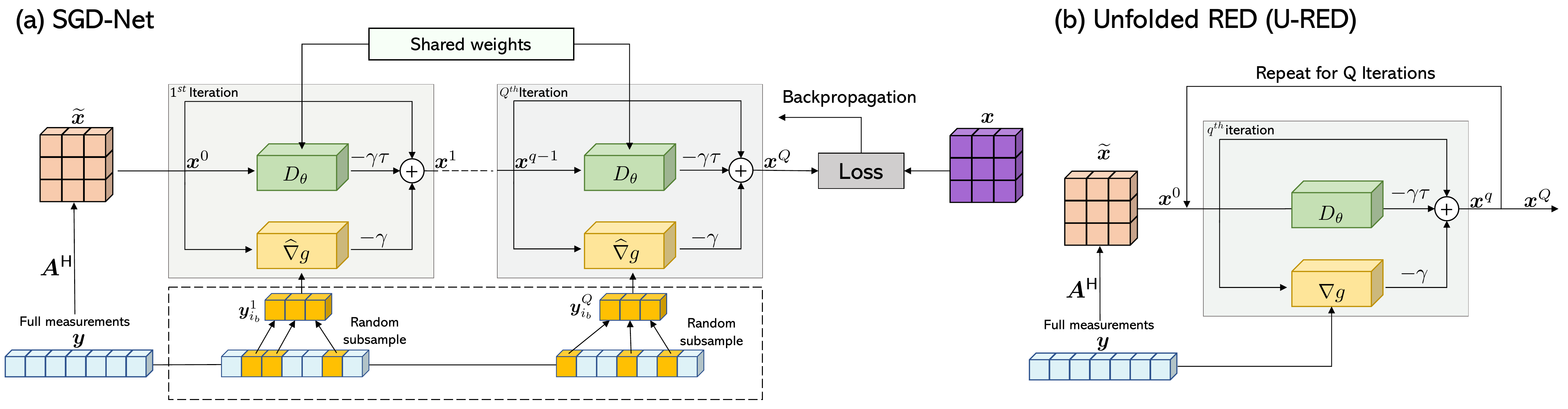}
\end{center}
\caption{Architecture of (a) \emph{SGD-Net} with $Q$ steps and (b) its batch variant that we label as \emph{Unfolded RED (U-RED)}. Data-consistency layers within SGD-Net rely on minibatch gradients with $B \ll I$ measurements, while those in U-RED rely on the full batch gradient using $I$ measurements. The learned operator blocks within SGD-Net and U-RED share their weights across steps and are trained in an end-to-end fashion by accounting for the data-consistency layers.}
\label{Fig:SCREDNET}
\end{figure*}

\section{Proposed Method}
\label{Sec:Proposed}
We now introduce SGD-Net that adopts stochastic approcessing of measurements within data-consistency layers. As corroborated by our results in Section~\ref{Sec:Experiments}, SGD-Net is ideal for data-intensive applications where the object features are difficult to characterize by a pre-trained CNN denoiser.

\subsection{Stochastic data-consistency layers}

Consider a data-fidelity  $g$ that consists of a set of $I \geq 1$ component functions $\{g_i(\xbm)\}$, where each $g_i$ depends on some subset $\ybm_i$ of the measurements in $\ybm$. For example, in tomographic imaging each $\ybm_i$ corresponds to a single projection of an object along a specific angle~\cite{Kak.Slaney1988}.  When $g$ corresponds to the least-squares penalty, the data-consistency layer within the deep unfolding network becomes
\begin{align}
\label{Eq:sgd_grad}
\nabla g(\xbm) = \frac{1}{I}\sum_{i=1}^I \nabla g_i(\xbm) = \frac{1}{I}\sum_{i=1}^I \Abm^{\Hsf}_i(\Abm_i\xbm - \ybm_i),
\end{align}
where $(\cdot)^{\Hsf}$ denotes the conjugate transpose operation. Note that the complexity of $\nabla g$ scales \emph{linearly} with the total number of components $I$. This means that when $I \rightarrow \infty$, the memory requirements or computation time of the traditional batch deep unfolding becomes impractical.  The central idea of SGD-Net, summarized in Fig~\ref{Fig:SCREDNET}, is to approximate the gradient at every step within the deep unfolding network with an average of $1 \leq B\ll I$ component gradients, which makes SGD-Net independent of the total number of components $I$. The corresponding \emph{minibatch} gradient is computed as
\begin{align}
\label{Eq:minibc_grad1}
\nablahat g(\xbm) =\frac{1}{B}\sum_{b=1}^B \nabla g_{i_b}(\xbm)  = \frac{1}{B}\sum_{b=1}^{B}\Abm^{\Hsf}_{i_b}(\Abm_{i_b}\xbm - \ybm_{i_b}),
\end{align}
where $i_1,\dots,i_B$ are independent random indices that are selected uniformly from the set $\{1,\dots,I\}$. The minibatch size parameter $B \geq 1$ controls the number of gradient components at each step of SGD-Net. Note that~\eqref{Eq:minibc_grad1} directly implies that $\E[\nablahat g(\xbm)] = \nabla g(\xbm)$, where the expectation is taken over the random indices $i_1, \dots, i_B$.

We will use the vector $\phibm$ to denote the physical parameters within the data-consistency layers of the unfolded network. In the context of the batch network in Fig.~\ref{Fig:SCREDNET}(b), which uses all the measurements at every step, the physical parameters correspond to the gradients $\{\nabla g_i(\xbm)\}$ at every step of the unfolded network. SGD-Net seeks to minimize the complexity of the unfolded network by replacing $\phibm$ with its minibatch approximation $\phihat$, obtained by applying~\eqref{Eq:minibc_grad1} to every step.

\subsection{Stochastic deep unfolding network}

Given an initial solution $\xbmtilde = \Abm^\Hsf\ybm$, SGD-Net iteratively refines it by infusing information from both the  minibatch gradient of the data-fidelity term $\nablahat g$ and the learned operator $D_{\thetabm}$ defined as
\begin{equation}
\label{Eq:artifact estimator}
D_{\thetabm}(\xbm) = (I - R_{\thetabm})(\xbm) = \xbm - R_\thetabm(\xbm),
\end{equation}
where $R_\thetabm$ is the artifact-removal CNN. Unlike in SIMBA~\cite{Wu.etal2020}, the prior in SGD-Net is optimized end-to-end using the training data to maximally reduce the artifacts. We fix the total number of SGD-Net steps to $Q \geq 1$, with each step given by
\begin{align}
\label{Eq:RED-Netupdate}
&\xbm^{q+1} =  \xbm^{q} - \gamma (\nablahat g(\xbm^{q}) + \tau D_\thetabm(\xbm^{q})),\quad\quad
\end{align}
where $\gamma>0$ is a step-size parameter.
Fig.~\ref{Fig:SCREDNET}(a) illustrates the algorithmic details of SGD-Net, which can, in principle, be implemented with or without weight-sharing across the $Q$ steps. In our implementation, we opted to share the weights of $R_{\thetabm}$ accross different steps to make it more suitable for large-scale imaging applications. While there are no trainable parameters within the stochastic data-consistency layers, one can still use backpropagation to compute the gradient of SGD-Net with respect to the trainable parameters $\thetabm$. Since SGD-Net randomly processes a subset of measurements at each step, the prediction of SGD-Net is in fact randomized. Our theoretical analysis in Section~\ref{Sec:Theory} precisely characterizes the training of SGD-Net relative to that of the batch network that uses all the measurement in every step.

Inspired by~\cite{zhang.etal2020a}, our CNN regularizer is based on the widely used U-Net architecture~\cite{Ronneberger.etal2015}. The corresponding CNN consists of four scales, each with a skip connection between downsampling and upsampling. These connection increase the effective receptive field of the network as the input goes deeper in the network~\cite{DJin.etal2017}. The number of channels in each layer are 32, 64, 128, and 256. We make two additional modifications to the U-Net. First, the activation function in our setting is PReLU (parametric ReLU, $f(x)=\max(0,x) + a*\min(0,x)$, where $a$ is a trainable parameter). Second, since we adopt small minibatch training, we use group normalization (GN)~\cite{Wu.etal2018} as an alternative to batch normalization (BN). The computation in GN is independent of the minibatch dimension, which makes its accuracy stable for a wide range of minibatch sizes.

\subsection{End-to-end training of SGD-Net}

\begin{figure*}[t]
\begin{center}
\includegraphics[width=16cm]{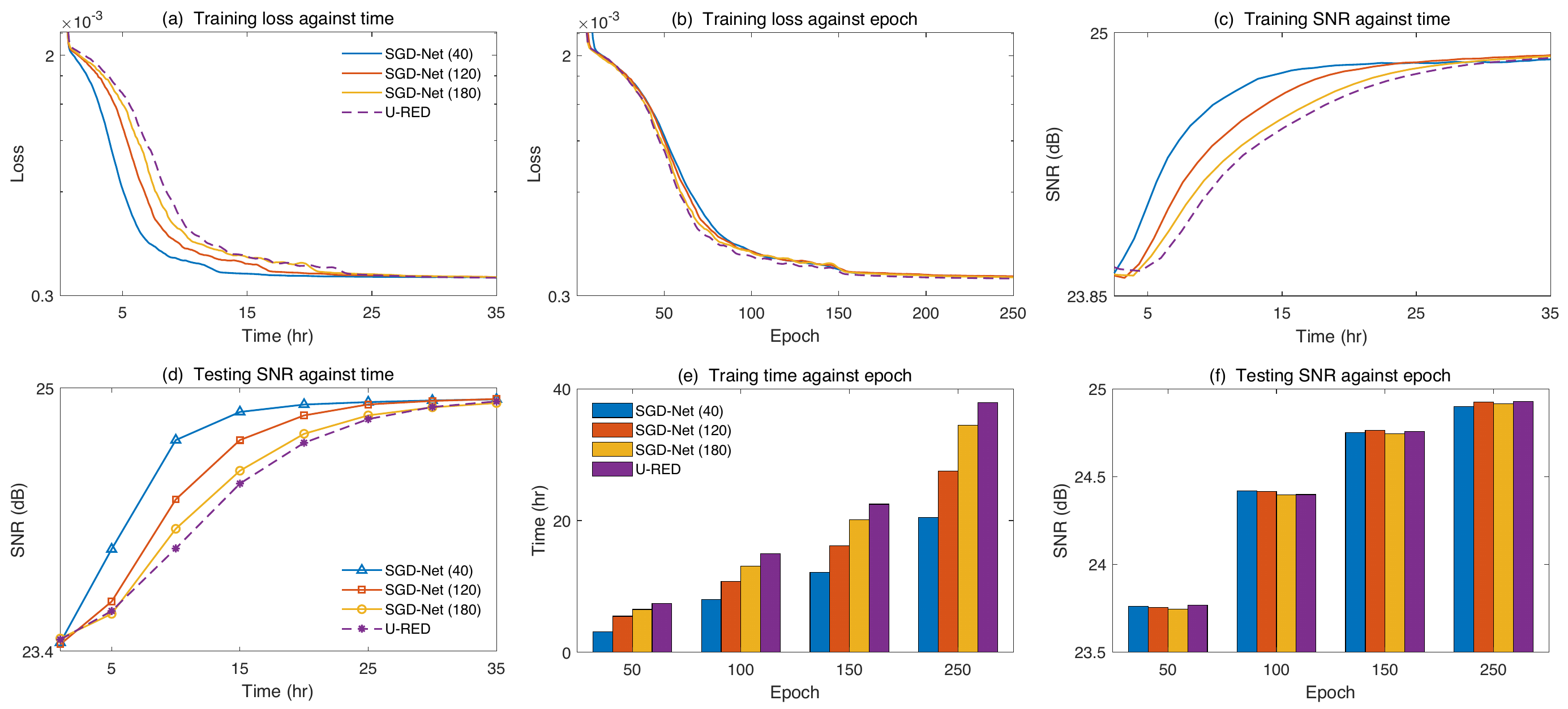}
\end{center}
\caption{Quantitative evaluation of SGD-Net on IDT for different minibatch sizes $B \in \{40, 120, 180\}$ used at each step of the network against U-RED using the full batch of $I = 240$ measurements. (a) Illustration of the loss against time in hours for different values of $B$ evaluated on the training set. (b) Illustration of the loss against the epoch number  evaluated on the training set. (c) Illustration of the SNR (dB) against time evaluated on the training set. (d) Illustration of the SNR (dB) against time in hours evaluated on the testing set. (e) Illustration of the amount of time required to reach a certain epoch for different values of $B$. (f) Illustration of the SNR (dB) achieved at different epochs for different values of $B$ evaluated over the testing set. The figure highlights that by using minibatches of size $1 \leq B \ll I$ one can achieve nearly $2 \times$ improvement in training time over U-RED for the same final imaging quality.}
\label{Fig:charplot}
\end{figure*}

Let $\Tcal_{\thetabm;\phibm}(\xbmtilde, \ybm)$ be the batch \emph{unfolded RED (U-RED)} network, where $\thetabm$ are the \emph{learnable} parameters within $D_\thetabm$, $\phibm$ are \emph{physical} parameters within the data consistency layers, $\ybm$ is the measurement vector, and $\xbmtilde$ is the input to the network. End-to-end training seeks to compute the learnable parameters $\thetabm$ of $\Tcal_{\thetabm;\phibm}$ by minimizing a loss function $F$ over $M$ training samples $\{\xbm_j,\ybm_j\}$. Let $F_j$ be a loss functions over the training sample $(\xbm_j, \ybm_j)$, then the training is formulated as the following optimization problem over learnable parameters
\begin{equation}
\label{EqRED-Netobj}
\argmin_{\thetabm}\;\, F(\thetabm;\:\bm\phi)\quad\text{with}\;\; F(\thetabm;\:\bm\phi) \defn \frac{1}{M}\sum_{j=1}^{M}F_j(\thetabm;\:\phibm).
\end{equation}
For example, a popular choice for the loss function $F_j$ as the \emph{mean square error (MSE)} between the image predicted by $\Tcal_{\thetabm, \phibm}$ and the desired image $\xbm_j$
\begin{equation}
\label{EqRED-Netupdate_exp}
F_j(\thetabm:\phibm) = \left\|\Tcal_{\thetabm,\phibm}(\xbmtilde_j, \ybm_j) - \xbm_j\right\|_2^2 \text{ with }\xbmtilde_j= \Abm^\Hsf\ybm_j.
\end{equation}

One can train both U-RED and SGD-Net using gradient-based optimizers, such as the \emph{stochastic gradient descent (SGD)}~\cite{Robbins.Monro1951}. When training SGD-Net, iteration $k$ of training is performed by generating two sets of independent random variables. First, an index in $j_k $ is randomly selected from $\{1, \dots, M\}$, then, a stochastic approximation $\phihat^k$ of $\phibm$ is generated by replacing the batch gradients by their randomized minibatch approximations~\eqref{Eq:minibc_grad1} at every step of the network. The trainable weights can then be updated as
\begin{equation}
\label{EqSGD-update}
\thetabm^{k+1} = \thetabm^{k} - \eta_{k}\nablahat F(\thetabm^{k}; \phihat^k),
\end{equation}
where $\nablahat F(\thetabm^{k}; \phihat^k)  = \nabla F_{j_k}(\thetabm^{k}; \phihat^k)$ and $\eta_{k} > 0$ is the step-size at the training iteration $k \geq 0$. As illustrated in Fig.~\ref{Fig:SCREDNET}, SGD-Net can significantly reduce the complexity of the data-consistency layers by using $\phihat^k$ instead of the full $\phibm$. In Section~\ref{Sec:Theory}, we present a theoretical analysis of training SGD-Net using the SGD iteration~\eqref{EqSGD-update}.

We use a warm up strategy to initialize the network for training. We first train the artifact-removal CNN $R_{\thetabm}$ separately with the MSE loss in~\eqref{Eq:cnn_loss} after initializing the parameters with random values. This training is considerably faster than training the entire network since it has no data-consistency blocks. Since we used a recursive network with the same weights across iterations, the weights of the unfolded network at each iteration are initialized using the weights learned from the pre-trained network $R_{\thetabm}$.

\section{Theoretical analysis}
\label{Sec:Theory}

We now present the theoretical analysis of the training of SGD-Net by using the SGD algorithm in~\eqref{EqRED-Netupdate_exp}. Note that our analysis does not explicitly assume that the unfolded architecture corresponds to RED, which means that it is also applicable to other architectures, including those based on PnP.

Our theoretical analysis requires a number of technical assumptions that act as sufficient conditions for the main theorem below. Our first assumption is on Lipschitz continuity of the gradient of $F$ over both $\phibm$ and $\thetabm$.
\begin{assumption}
\label{As:LipschitzCon}
The function $F$ has a global minimizer $\thetabm^\ast$ and satisfies the following continuity assumption
\begin{equation*}
\begin{aligned}
\|\nabla &F(\thetabm_1;\phibm_1) -\nabla F(\thetabm_2;\phibm_2) \|_2\\
&\leq L_{\theta}\|\thetabm_1 - \thetabm_2\|_2 + L_{\phi}\|\phibm_1 - \phibm_2\|_2,
\end{aligned}
\end{equation*}
with $L_\theta, L_\phi > 0$ for every $(\thetabm_1, \phibm_1)$ and $(\thetabm_2, \phibm_2)$.
\end{assumption}
\noindent
The existence of a minimizer, and the Lipschitz continuity of the gradient are standard assumptions in traditional optimization~\cite{Nesterov2004}. Assumption~\ref{As:LipschitzCon} simply extends Lipschitz continuity over both sets of parameters within the cost function $F$. Note that we do \emph{not} assume that $F$ is convex.

A common assumption in stochastic optimization is that the minibatch gradients are unbiased estimators of the full gradient and have bounded variances~\cite{Ghadimi.Lan2016, Wu.etal2019}. Our second assumption states this condition for the physical parameters of the unfolded network.

\begin{assumption}
\label{As:UnbiasedAssumption}
The physical parameters of the network satisfy
\begin{equation*}
\E\left[\phihat^k\right]\,=\,\phibm, \quad\E\left[\|\phihat^k - \phibm\|_2^2\right]\,\leq\,\frac{\sigma^2}{B},
\end{equation*}
for all iterations $k \geq 0$, where $\sigma>0$ is a constant and $B \geq 1$ is the minibatch size.
\end{assumption}
\noindent
This is a mild assumption since in our SGD-Net implementation, $\phihat$ is obtained by replacing the full gradient $\nabla g$ at every step by its minibatch gradient $\nablahat g$ computed via~\eqref{Eq:minibc_grad1}. This automatically ensures that $\phihat$ is an unbiased estimator of $\phibm$. Our final assumption is related to Assumption~\ref{As:UnbiasedAssumption}, but considers the selection of the gradients $\nablahat F$ during SGD training.
\begin{assumption}
\label{As:StochCost}
The stochastic gradients in~\eqref{EqSGD-update} satisfy the following two conditions for any fixed vectors $\thetabm$ and $\phibm$.
\begin{enumerate} \renewcommand{\labelenumii}{(\arabic{enumii})}
     \item[(a)] The stochastic gradient is unbiased:  \[\E\left[\nablahat F(\thetabm;\phibm)\right]\,=\,\nabla F(\thetabm;\phibm) . \]
     \item[(b)]  The variance of the stochastic gradient is bounded: \[\E\left[\|\nablahat F(\thetabm; \phibm)- \nabla F(\thetabm;\phibm)\|_2^2\right]\,\leq\, \epsilon^2.\]
\end{enumerate}
The expectations  are taken with respect to the random index $j \in \{1, \dots, M\}$ used to select the training sample.
\end{assumption}
\noindent

\begin{theorem}
\label{Thm:MainConvRes}
Run the SGD learning in~\eqref{EqSGD-update} for $K \geq 1$ iterations under Assumptions~\ref{As:LipschitzCon}-\ref{As:StochCost} using the step-size parameters $1 \leq \eta_{k} \leq 1/L_\theta$  and the minibatch size $B \geq 1$. Then, the iterates generated by~\eqref{EqSGD-update} satisfy the bound
\begin{equation*}
\begin{aligned}
\label{EqTheorem1}
&\sum_{k=0}^{K-1}\eta_{k} \E\left[ \|\nabla F(\thetabm^{k};\,\phibm) \|_2^2\right] \leq 2(F(\thetabm^0; \phibm) - F(\thetabm^\ast; \phibm)) + \frac{L_\phi^2\sigma^2}{B}\left(\sum_{k=0}^{K-1} \eta_{k}\right)+L_{\theta}\epsilon^2 \left(\sum_{k = 0}^{K-1}\eta_{k}^2\right).
\end{aligned}
\end{equation*}
\end{theorem}
\begin{proof}
See the appendix.
\end{proof}
Theorem~\ref{Thm:MainConvRes} allows us to establish various forms of convergence results by controlling the step-sizes $\eta_k$. For example, when $\eta_k  = 1/(L_\theta \sqrt{K})$, one obtains
$$\min_{k \in \{0,\dots,K-1\}} \E\left[\|\nabla F(\thetabm^k; \phibm)\|\right] \leq \frac{C}{\sqrt{K}} + \frac{L_\phi^2 \sigma^2}{B},$$
where $C = 2(F(\thetabm^0; \phibm)-F(\thetabm^\ast; \phibm)) L_\theta + \epsilon^2$ is a constant. This implies that SGD in~\eqref{EqSGD-update}, which relies only on the minibatch approximation $\phihat^k$ of $\phibm$, achieves, in expectation, the \emph{first-order necessary conditions} of optimality for~\eqref{EqRED-Netobj} up to an error term $L_\phi^2 \sigma^2/B$. This error term can be made as small as possible by controlling $B$ within SGD-Net. SGD-Net can thus be trained to approximate the batch deep unfolding network that uses $\phibm$ to any desired precision. Section~\ref{Sec:Experiments} validates this agreement for different values of $B$ by also showing substantial computational savings in training and testing.

\section{Numerical Validation}
\label{Sec:Experiments}
We now empirically validate SGD-Net in the context of two computational imaging modalities, IDT and sparse-view CT.  Our first goal is to validate the proposed theorems in Section~\ref{Sec:Theory} and the second one is to highlight the effectiveness and efficiency of our method for processing a large number of measurements. All the experiments were performed on a machine equipped with an Intel Xeon Gold 6130 Processor and eight NVIDIA GeForce RTX 2080 Ti GPUs.
\begin{figure*} [t!]
    \centering
    \includegraphics[width=16cm]{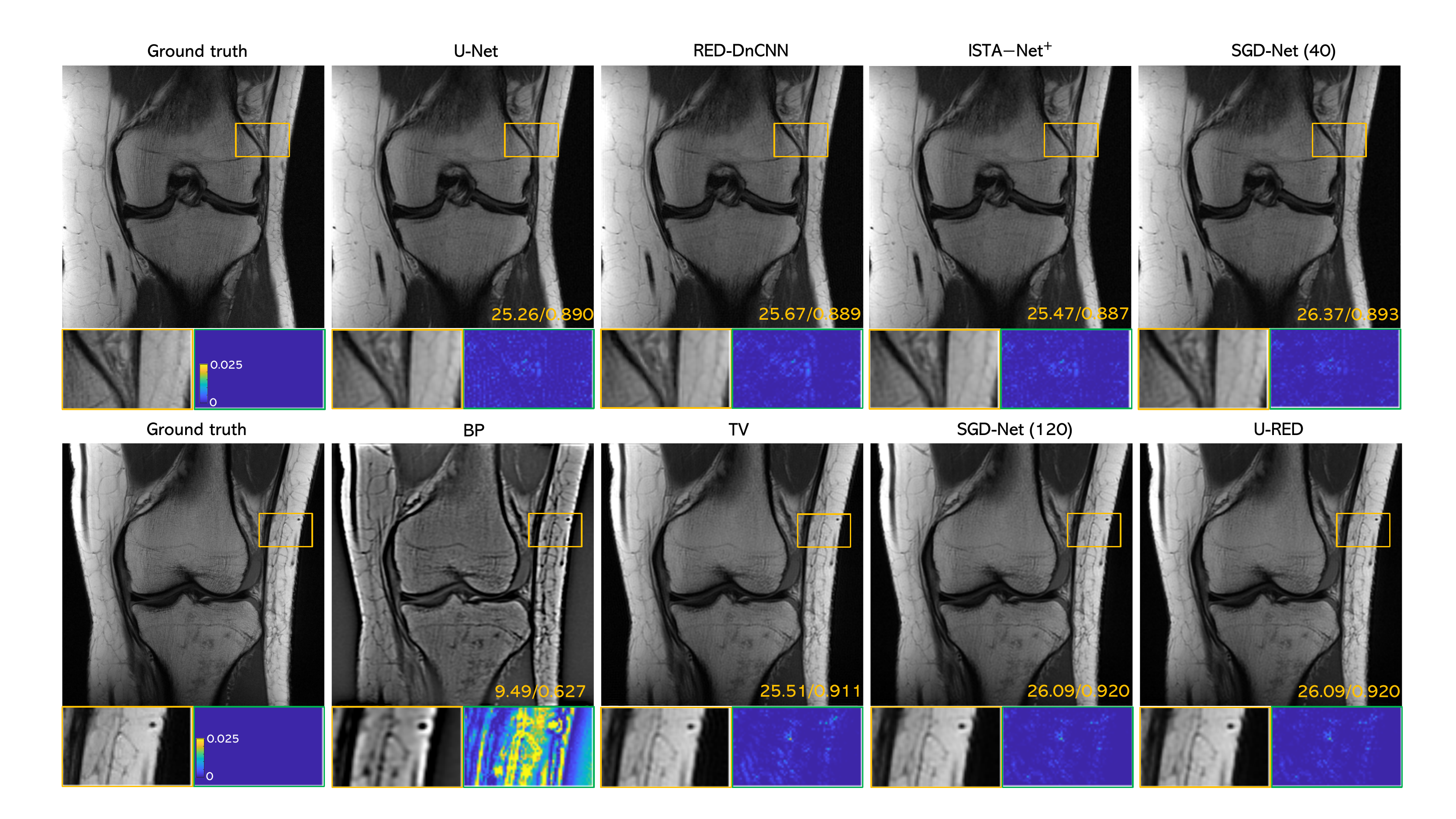}
    \renewcommand\arraystretch{1.3}
    \caption{Quantitative evaluation of several well-known methods on IDT under noise corresponding to input SNR of 20 dB. The total number of IDT measurements in this experiment is $I = 240$. All the baseline methods use the full set of measurements $I$ during reconstruction. U-RED corresponds to the full batch architecture that uses all the measurements at every step. SGD-Net (40) and SGD-Net (120) use minibatches of size $B = 40$ and $B = 120$, respectively, at every step.  Each image is labeled with its SNR (dB) and SSIM values with respect to the original image. The yellow box provides a close-up with a corresponding error map provided on its right.  The results highlight competitive performance of SGD-Net relative to several well-known methods, while also showing its ability to match the imaging quality achieved by the batch U-RED network.}
    \label{Fig:FigIDT}
\end{figure*}
Several image reconstruction methods were used as references, including TV~\cite{Rudin.etal1992}, U-Net~\cite{Ronneberger.etal2015}, RED-DnCNN~\cite{Romano.etal2017} and ISTA-Net$^{+}$~\cite{zhang2018ista}.  TV is formulated in~\eqref{Eq:ForwardModel}, and was implemented using the accelerated proximal gradient descent method (APGM)~\cite{Beck.Teboulle2009a}. U-Net corresponds to our own implementation of the architecture used in~\cite{Sun.etal2018}. The network was trained in the usual supervised fashion using the $\ell_2$-loss~\cite{DJin.etal2017}. We adopted the DnCNN architecture used in~\cite{Metzler.etal2018} as the AWGN denoiser for RED.  The network has seventeen layers, including 15 hidden layers, an input layer, and an output layer. We have also experimented using RED with the U-Net like architecture from SGD-Net, but observed that this does not improve the image reconstruction quality achieved by RED. ISTA-Net$^+$ is a widely-used deep unfolding architecture based on the feed-forward network obtained by unfolding and truncating ISTA. We unfolded ISTA-Net$^+$~\footnote{The code for ISTA-Net$^+$ is publicly available at \url{https://github.com/jianzhangcs/ISTA-Net-PyTorch}.} for 12 steps and set the number of feature maps in each convolution layers equal to 64 for the best SNR performance in our experimental settings. Note that \emph{U-RED} uses the complete set of measurements $I$ corresponds to the traditional deep unfolding of the batch RED algorithm. All methods were implemented in Pytorch with a GPU backend.

 \begin{table*}
\caption{SNR and SSIM values obtained by several reference methods for IDT. Note that the last two columns provide the GPU memory usage and the run-times for all the methods for reconstructing a 320$\times$320 image. The total number of measurement is $I=240$. Note the excellent balance between quality and complexity achieved by SGD-Net.}
    \centering
    \renewcommand\arraystretch{1.2}
    {\footnotesize
    \scalebox{0.8}{
    \begin{tabular}{|c|c|c|c|c|c|c|c|c|c|}
        \hline
        \bf Metric       & \multicolumn{3}{c|}{\bf SNR}                        & \multicolumn{3}{c|}{\bf SSIM}         &\multirow{3}{7em}{\centering \bf\#Iterations} & \multirow{3}{9em}{\centering \textbf{Size}\\Model/Measurement}& \multirow{3}{7em}{ \centering  \textbf{Time}  CPU/GPU}              \\  \cline{1-7}
        \diagbox[width=3.2cm, height=1.3cm]{\bf Method}{\bf Input-SNR\\ \bf(dB)} &20-5 &20 & 20+5 & 20-5 &  20 & 20+5 & & &\\ \hline\hline
        TV       & 24.26                  & 24.31                  & 24.39                  & 0.887   & 0.890                 & 0.891    & \centering 250 & --------/1.01$\,$GB   & 87.58s/10.66s                \\
        U-Net        & 24.27                  & 24.33                  & 24.35                  & 0.887     & 0.889                  & 0.889   & \centering --  & 118.2$\,$MB/--------&  0.925s/0.012s           \\
        ISTA-Net$^{+}$       & 24.39                  & 24.41                  & 24.47            & {0.889}     & 0.890                  & 0.890  & \centering 12 & 6.9$\,$MB/1.01$\,$GB & 18.36s/0.402s      \\
        RED-DnCNN              & 24.54                  & 24.61                  & 24.67                  & 0.890     & 0.892                  & 0.893   & \centering 220       & 2.29$\,$MB/1.01$\,$GB   & 197.5s/4.144s             \\
        SGD-Net (40)      & {24.84}            & {24.94}            & {24.96}            & {0.896}     & 0.899                  & 0.901  & \centering 8  & 29.6$\,$MB/0.17$\,$GB & 7.443s/0.322s          \\
        SGD-Net (120)      & {24.87}            & {24.93}            & {24.94}            & {0.898}     & 0.899                  & 0.900    & \centering 8 & 29.6$\,$MB/0.51$\,$GB  & 16.51s/0.617s         \\
        U-RED       & {24.89}            & {24.93}            & {24.94}            & {0.898}         & 0.899                  & 0.900   & \centering 8  & 29.6$\,$MB/1.01$\,$GB & 31.23s/0.943s       \\ \hline
    \end{tabular}}
    }
    \label{Tab:table1}
\end{table*}

We used the following \emph{signal-to-noise ratio (SNR)} in dB for quantitively comparing different algorithms
\begin{equation}
\label{Eq:SNR}
\hbox{SNR}(\xbfhat,\xbm) = \max_{a,b\in \R} \left\{20 \log_{10}\left( \frac{\|\xbm\|_2}{\|\xbm - a\xbfhat + b\|_2}\right)  \right\},
\end{equation}
where $\xbfhat$ and $\xbm$ represents the noisy vector and ground truth respectively, while the purpose of $a$ and $b$ is to adjust for contrast and offset. We also used the \emph{structural similarity index measure (SSIM)}~\cite{Wang.etal2004} as an alternative metric.

\subsection{Intensity diffraction tomography}
IDT~\cite{Ling.etal18} is a data intensive computational imaging modality that seeks to recover the spatial distribution of the complex
permittivity contrast of an object given a set of its intensity-only measurements. In this problem, $\Abm$ consists of a set of $I$ complex measurement operators $[\Abm_1, \dots, \Abm_I]^{\Tsf}$, where each $\Abm_i$ is a convolution corresponding to the $i$th measurement $\ybm_i$. In the simulation, we extracted a random subset of 350 slices of 320$\times$320 images for training, 10 images for validation, and 35 images for testing from the NYU fastMRI Initiative database~\cite{knoll2020fastmri}. Followed by the experimental setup in~\cite{Ling.etal18, Wu.etal2020}, the simulated images are assumed to be on the focal plane $z\,=\,0\SI{}{\micro\metre}$ with LEDs located at $z_{\text{LED}} = -70 \SI{}{\milli\metre}$. The wavelength of the illumination was set to $\lambda = 630\SI{}{\nano\metre}$ and the background medium index was assumed to be water with $\epsilon_b = 1.33$.  We generated $I = \,$240 intensity measurements with $40\times$ microscope objectives (MO) and 0.65 numerical aperture (NA). All simulated measurements were additionally corrupted by AWGN corresponding to $\{15, 20, 25\}$ dB of input SNR.
\begin{figure*}[t!]
\begin{center}
\includegraphics[width=16cm]{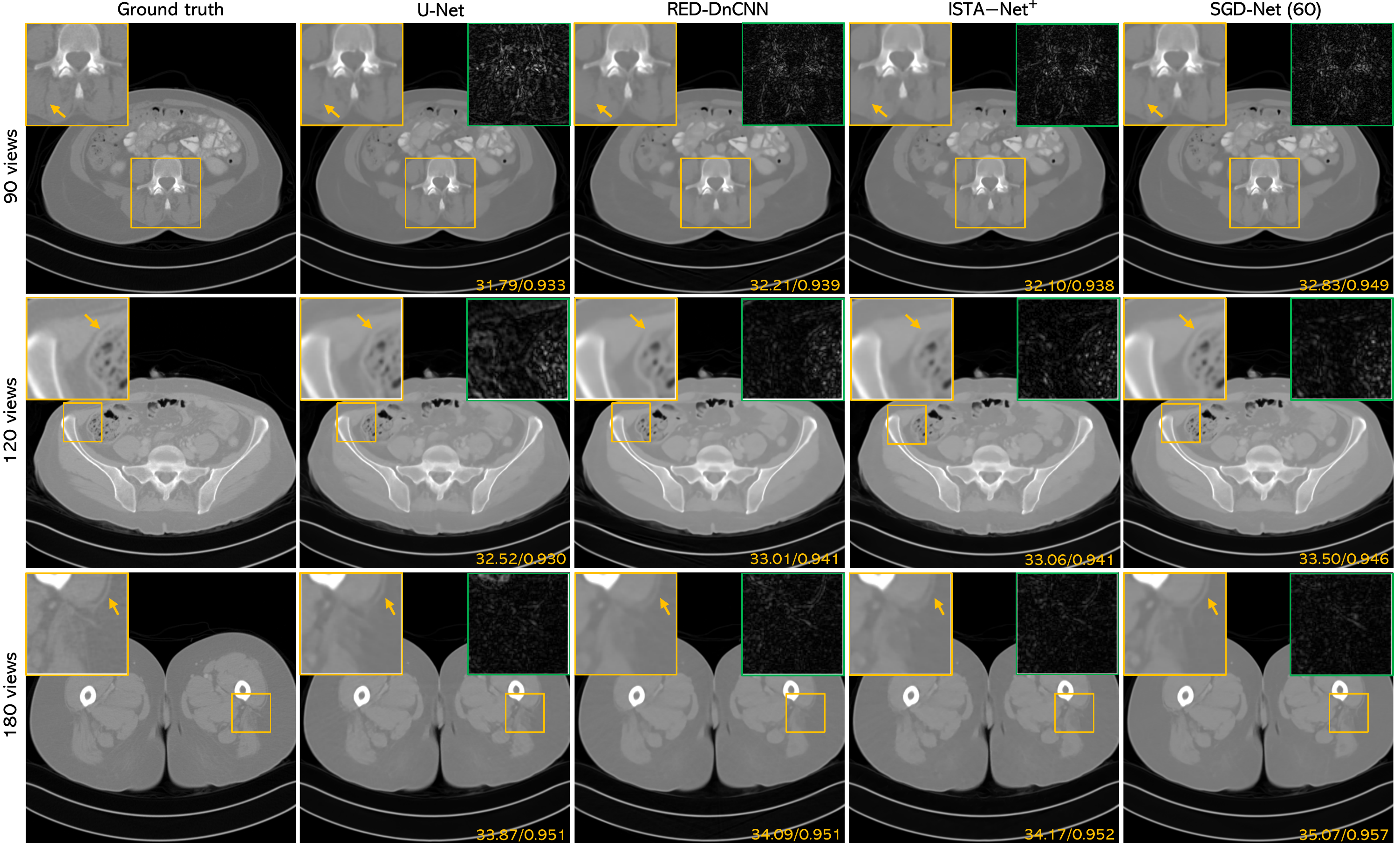}
\end{center}
\caption{Visual illustration of SGD-Net with $B=60$ relative to several well-known deep-learning baseline methods on sparse-view CT with $\{90, 120, 180\}$ views and noise of 50 dB input SNR. Note that RED-DnCNN and ISTA-Net$^{+}$ use the full set of measurements at each step. Each image is labeled with the corresponding SNR (dB) and SSIM values. This figure highlights that SGD-Net can achieve competitive performance relative to several well-known methods, while also providing a mechanism to reduce the complexity of data-consistency layers.}
\label{Fig:Fig4_v2}
\end{figure*}

We trained SGD-Net with the initialization $\xbm^0 = \Abm^{\Hsf}\ybm$, where $\Abm^{\Hsf}$ denotes the conjugate transpose. The proposed recursive model was unrolled for $Q=8$ steps and trained using SGD with minibatch size 1. While the step-size parameter of the data-consistency block in each step of SGD-Net was fixed to $\gamma=5\times10^{-3}$, the regularization parameter was set as a learnable parameter, initialized with $\tau=2$. The learning rate of SGD was set in two stages. In the first 150 epochs, we adopt the cyclic learning rate policy~\cite{Smith2017}, where the policy cycles the learning rate between $8\times10^{-3}$ and $4\times10^{-3}$ with 2000 training iterations in the decreasing half of a cycle. In stage 2, the learning rate was gradually reduced by a factor of 0.5 every 50 epochs. The number of total training epochs was 300. It is worth to note that each SGD-Net was trained with the noise corresponding to 20 dB of input SNR in order to test the stability of the proposed method with respect to changes in amount of measurements noise. Both TV and RED-DnCNN were initialized with $\xbm=\zerobm$, and we performed grid search to identify the optimal regularization parameters. For the DnCNN denoiser in RED, we trained it for AWGN removal at four noise levels corresponding to $\sigma\in$\{ 5, 10, 15, 20\}. For each experiment, we selected the denoiser achieving the highest SNR value. Several instances of U-Net and ISTA-Net$^+$ were trained by mapping the \emph{backprojection} (BP) $\Abm^{\Hsf}\ybm$  to the ground truth for each input SNR levels. We initialized the step-size and the regularization parameters of ISTA-Net$^+$ to the same values as SGD-Net, subsequently learned all these parameters during training, as done in the original paper.

Fig.~\ref{Fig:charplot} highlights the ability of SGD-Net to reduce the computational complexity of training compared to the full batch network U-RED. The three top plots compares the average loss and SNR achieved by SGD-Net when evaluated on the training set using different values for $B$ at each step. Note that U-RED uses a fixed set of full ($I=240$) measurements, while SGD-Net selects a random subset of $B$ measurements at every step. Fig.~\ref{Fig:charplot}(d) presents the average SNR achieved by SGD-Net against the training time evaluated on the testing set. Fig.~\ref{Fig:charplot}(e) highlights the time necessary to run a fixed number of epochs for different values of $B$. Fig.~\ref{Fig:charplot}(f) shows the SNR achieved by SGD-Net for different epochs. Note that the average time of training 250 epochs of SGD-Net using $B \in \{40, 120, 180\}$ and U-RED was 20.50 hours, 27.55 hours, 34.51, hours and 37.90 hours, respectively.  We did not observe significant SNR differences when using smaller minibatches for training compared with the usage of the full measurements. This highlights the ability of SGD-Net to reduce the complexity in deep unfolding by maintaining excellent imaging quality.

\begin{table*}[t!]
 \centering
 \caption{SNR and SSIM values obtained by several reference methods for the reconstruction of a 512$\times$512 image in sparse-view CT  with noise of 50 dB input SNR. The highest SNR and SSIM values are in bold. Note that the last two rows provide the average test-times for all the competing methods for 180 views on GPU and CPU. SGD-Net enables one to balance the time complexity of reconstruction against the final imaging quality.}
    \scalebox{0.7}{
    \begin{tabular}{|c|c||cccccccc|}
        \hline
     \multirow{3}{5em}{\centering \textbf{Views}  } & \multirow{3}{5em}{\centering \textbf{Metric} } & \multicolumn{8}{c|}{\textbf{Method} } \\
     \cline{3-10}
     & &\multirow{2}{4.2em}{\centering  FBP  }  & \multirow{2}{4.2em}{\centering  TV  }  &\multirow{2}{4.2em}{ \centering U-Net  }  &\multirow{2}{4.2em}{\centering   RED-DnCNN }  & \multirow{2}{4.2em}{\centering ISTA-Net$^{+}$ }  & \multirow{2}{4.2em}{\centering SGD-Net\;(30)} & \multirow{2}{4.2em}{\centering SGD-Net\;(60)} &  \multirow{2}{4.2em}{\centering U-RED} \\
      & & & & & & & & &  \\ \hline \hline
        \multirow{2}{5em}{\centering 90}&SNR            & \centering17.56  & \centering 30.09 &\centering  31.17    & \centering 31.93 & \centering 32.01 & \centering 32.76 & \centering \bf 32.88  &32.87 \\
        &SSIM           & 0.362 & 0.924  & 0.930    & 0.935 & 0.934& 0.942& \bf 0.943 & \bf 0.943 \\ \hline
         \multirow{2}{5em}{\centering 120}&SNR            & 20.03  & 31.23 & 32.54    & 33.13 & 33.17 & 33.91 &33.95  & \bf 34.01 \\
        &SSIM           & 0.449 & 0.929 & 0.936    & 0.941 & 0.940 & 0.948& 0.949 &\bf 0.950  \\ \hline
         \multirow{2}{5em}{\centering 180}&SNR            & 23.19  & 32.97  & 34.04    & 34.49 & 34.61& 35.44 & \bf 35.46  & \bf 35.46 \\
        &SSIM           & 0.582 & 0.940 & 0.948    & 0.950 & 0.951 & 0.957 & \bf 0.958 & \bf 0.958  \\ \hline\hline
       \multirow{2}{7em}{\centering \textbf{Time}\\(views=180)}&CPU            & 0.859s  & 304.1s & 2.061s    & 460.3s & 15.95s & 11.56s & 13.31s  & 20.72s \\
        &GPU           & 0.147s & 13.58s & 0.217s    & 5.177s & 0.331s & 0.269s & 0.278 & 0.325  \\ \hline
    \end{tabular}}
\label{Tab:table2}
\end{table*}

\begin{figure}
\begin{center}
\includegraphics[width=10cm]{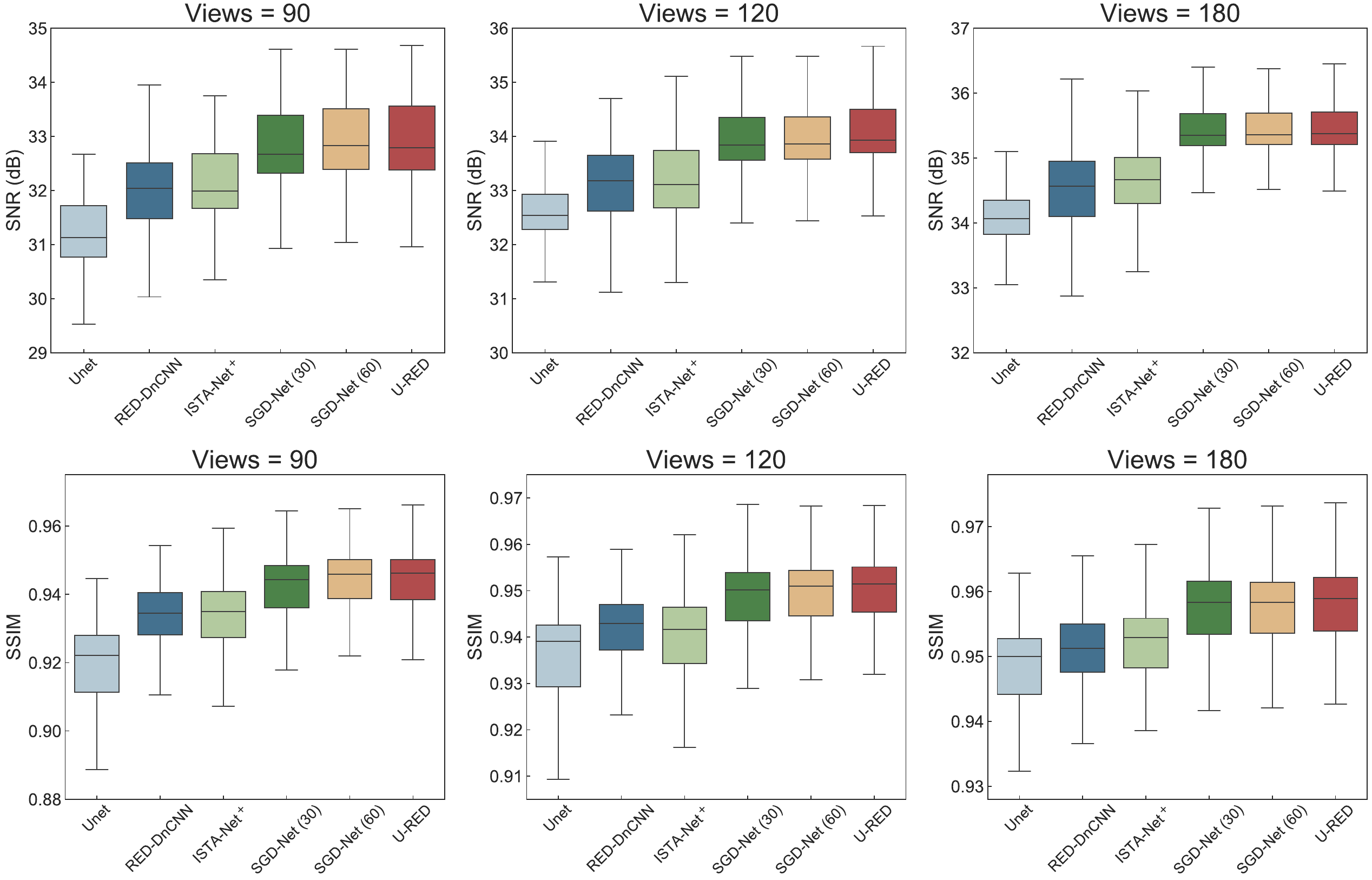}
\end{center}
\caption{Quantitative evaluation of several reference methods on sparse-view CT with noise corresponding to 50 dB of input SNR. The top row quantifies the image quality in terms of SNR (dB), while the bottom row in terms of SSIM. Columns from left to right provide results for 90, 120, and 180 views. This figure highlights that the usage of minibatches within SGD-Net does not reduce its ability to achieve high imaging quality.}
\label{Fig:boxplot}
\end{figure}

Table~\ref{Tab:table1} provides the final SNR and SSIM values achieved by SGD-Net and several baseline methods when applied to IDT at three noise levels. Overall, model-based deep learning methods, such as RED-DnCNN, ISTA-Net$^+$ and SGD-Net, achieve the best performances.  Moreover,  SGD-Net using $B  = 40$ and $B = 120$ match the performance of the batch algorithms in terms of the final reconstruction quality. The runtime in the table corresponds to the average inference time that excludes the model loading. Specifically, SGD-Net with $B=40$ is around $2.9\times$ faster than U-RED on GPU and around $4.2\times$ faster than U-RED on CPU with parallel processing. The memory column in the table corresponds to the usage of GPU memory. Specifically, the memory considerations in image reconstruction must take into account the size of all variables related to the image volume $\xbm$, the measured data $\{\ybm_i\}$, and the measurement operators $\{\Abm_i\}$. SGD-Net addresses the problems where the bottleneck is in the storage and processing of the measurements and measurement operators on GPU for the end-to-end training. Our implementation stores each $\Abm_i$ as two separate arrays for phase and absorption. In addition, each matrix is stored in the Fourier space to reduce the computational complexity of evaluating convolutions. This results in the storage of complex valued arrays for each, consisting of pairs of single precision floats for every element. The real and imaginary parts of each array were then separated into two input channels of SGD-Net. Thus, the shape of each measurements and measurement operators  in U-RED for reconstructing one slice is $1\times 320\times320\times240\times2$. A detailed discussion on the IDT forward model is available in~\cite{Ling.etal18, Wu.etal2020}. While U-RED requires $1.01$ GB of GPU memory due to its processing of all measurements in every iteration,  SGD-Net with $B=40$ requires only $0.17$ GB, which is about $1/6$th of the full volume. This highlights the potential of applying SGD-Net to large scale image reconstruction.

Fig.~\ref{Fig:FigIDT} provides some visual examples highlighting the imaging quality obtained by SGD-Net relative to several baseline methods. Specifically, the top row of Fig.~\ref{Fig:FigIDT} presents the results obtained by several learning-based methods.  As shown in the zoomed regions and the corresponding error maps, SGD-Net with $B=40$ outperforms all other methods, while the performance of U-Net is suboptimal due to its inability to leverage explicit data-consistency layers. The bottom row of Fig.~\ref{Fig:FigIDT} highlights the comparable quality obtained by SGD-Net with $B=120$ and that using all $I = 240$ measurements.

\subsection{Sparse-view CT}
Conventional CT requires many views for high-quality image reconstruction. In the following experiments, we explore the possibility of high-quality imaging when reducing the number of views in CT imaging. We consider reconstruction of simulated data obtained from the clinically realistic CT images provided by Mayo Clinic for the \emph{AAPM Low Dose CT grand Challenge}~\cite{mccollough2016tu}. The data from $7$ patients was used for training, one patient data for validation, and two patients' data for testing. This provides us with $2070$ slices of $512 \times 512$ images for training, $150$ slices of $512\times 512$ images for validation. The testing data consists of $275$ slices of $512\times512$ images. We implemented $\Abm$ and $\Abm^{\Hsf}$ with \texttt{RayTransform} in Operator Discretization Library (ODL)\footnote{The code for ODL is publicly available at \url{https://github.com/odlgroup/odl}.}~\cite{adler.etal2017}, which uses GPU accelerated \texttt{astra-gpu} backend~\cite{vanAarle.etal16}. In particular, the scanning geometry is a fan-beam source with $I \in \{90, 120, 180\}$ projection views equally distributed around $360^{\circ}$ and 1447 detector pixels.  The sinograms were generated by slightly perturbing the angles of views by a zero-mean AWGN with standard deviation of 0.003 degrees so as to make the experiments more realistic~\cite{Gupta2018}.  We  added Gaussian noise to the sinograms to make the input SNR equal to 50 dB.

\begin{figure*}[t!]
\begin{center}
\includegraphics[width=16cm]{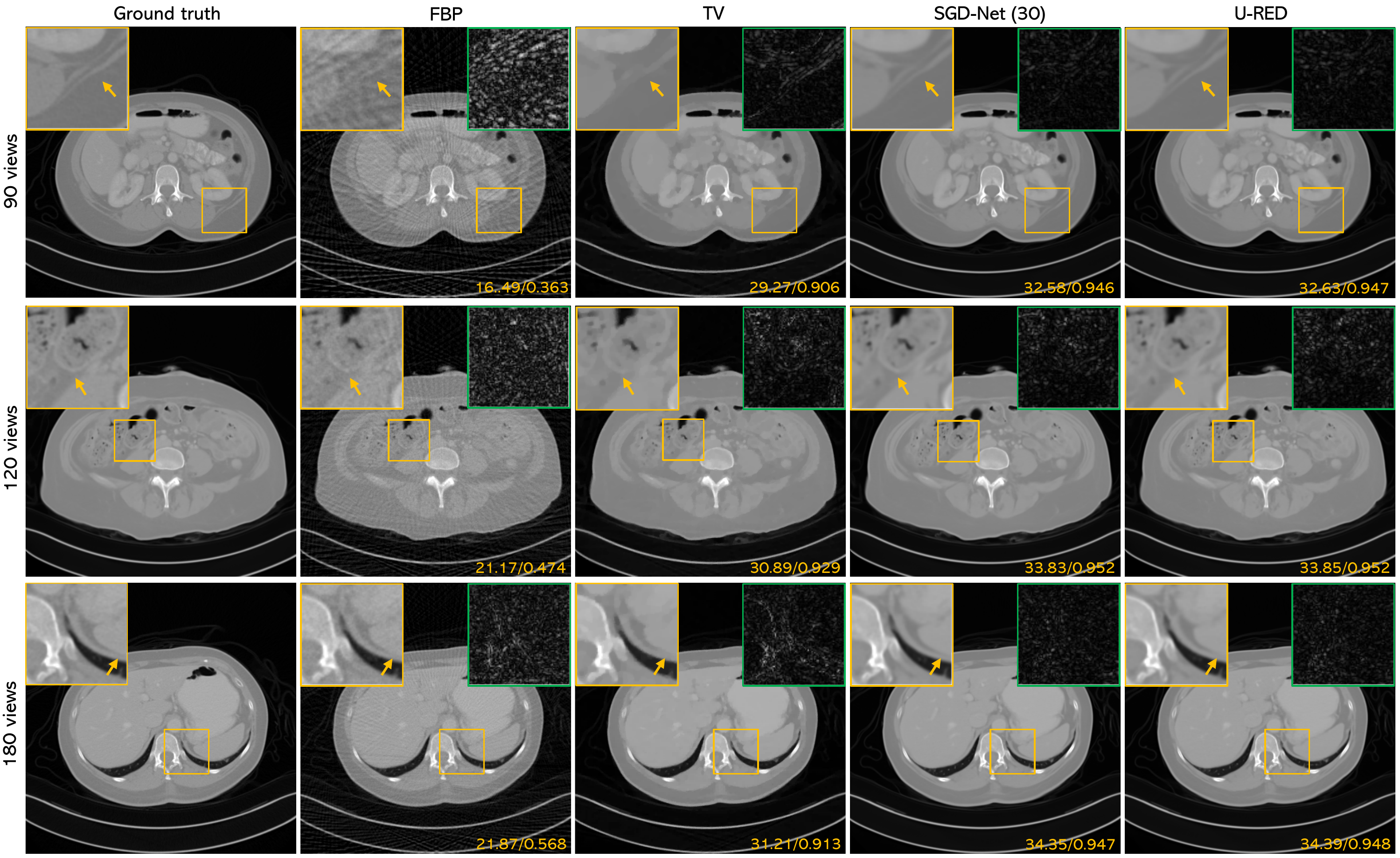}
\end{center}
\caption{Visual illustration of SGD-Net with $B=30$ relative to FBP, TV, and U-RED on sparse-view CT with $\{90, 120, 180\}$ views and noise of 50 dB input SNR.  See Table~\ref{Tab:table2} for quantitative and Fig.~\ref{Fig:Fig4_v2} for visual comparisons with additional reference methods. Note that SGD-Net provides substantial improvements over traditional image reconstruction methods, matching the performance of U-RED that uses all the measurements in each step of the network. }
\label{Fig:Fig5_v1}
\end{figure*}
\begin{figure*}[t!]
\begin{center}
\includegraphics[width=16cm]{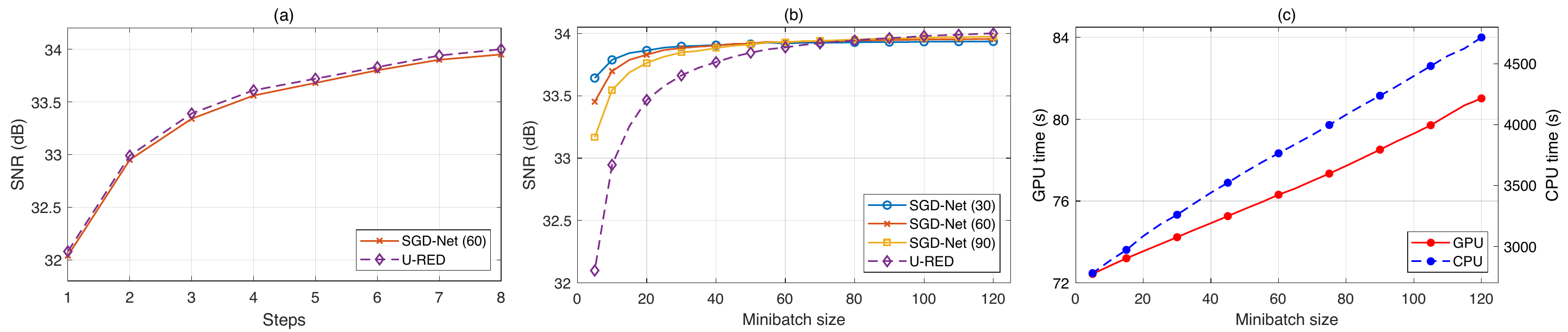}
\end{center}
\caption{Quantifying the influence of the number of steps and minibatch sizes. (a) SNR versus number of steps for the full batch network and SGD-Net with $B = 60$. Note how SGD-Net matches the full batch network in all settings. (b) SGD-Net architectures trained with certain minibatch sizes are tested on different minibatch sizes. This plot shows that networks trained on smaller minibatches achieve better generalization. (c) Runtime of SGD-Net for different minibatch sizes on GPU and CPU at test time. One can achieve significant savings for both GPU and CPU implementations with small minibatches.}
\label{Fig:bsSpeed}
\end{figure*}

We trained the SGD-Net by using the filtered backpojection (FBP) initialization $\xbm^0 = \Abm^{\Hsf}\Fbm\ybm$. The FBP is performed with a Hann filter by using the method \texttt{fbp-op} in ODL. We set the number of steps in SGD-Net to $Q=8$. For these experiments, we train SGD-Net using the ADAM solver~\cite{Kingma.Ba2015}, using the minibatch size $2$ and weight decay $2\times10^{-8}$. For physical parameters, we fixed the step-size in each step to $\gamma=5\times10^{-3}$ and initialized the trainable regularization parameter to $\tau=4$. The learning rate starts from $1\times10^{-3}$ and is halved at epoch 20, then gradually reduced by a factor of $0.7$ every $10$ epochs. The number of total training epochs is $100$. We adopt gradient clipping~\cite{Zhang2020Why} in order to accelerate and stabilize the training. We also applied the same training schemes for U-Net and ISTA-Net$^+$. We set the number of iterations for TV and RED to $240$. For DnCNN in RED, we trained it for AWGN removal at four noise levels corresponding to $\sigma\in \{2, 5, 10, 15\}$.  Several instances of U-Net and ISTA-Net$^+$, corresponding to different numbers of views, were trained by mapping the FBP  $\Abm^{\Hsf}\Fbm\ybm$  to the ground truth for each sparse view.

Table~\ref{Tab:table2} provides the SNR and SSIM values obtained by SGD-Net and the baseline methods when applied to sparse-view CT with different projection views.  Overall, all methods offer significant gains over FBP, with model-based deep learning methods (RED-DnCNN, ISTA-Net$^+$ and SGD-Net) achieving the best performance. Moreover,  SGD-Net using $B \in \{30, 60\}$ achieves comparable image reconstruction quality to U-RED, which highlights the potential of using minibatches within deep unfolding networks. Similar conclusions can be drawn by Fig.~\ref{Fig:boxplot}, which shows the statistical evaluation of SGD-Net and other deep-learning methods. Visual inspection of the results highlight the excellent performance of SGD-Net. Specifically, Fig.~\ref{Fig:Fig4_v2} presents the reconstruction results by all the learning-based methods.  As shown in the zoomed regions and the corresponding error maps, the average performance of U-Net is lower than all other methods since it does not explicitly leverage the data-consistency information. SGD-Net with $B=60$ performs better than other methods on all individual slices, highlighted by the yellow arrow. In Fig.~\ref{Fig:Fig5_v1},  FBP is dominated by streaking artifacts, while TV reduces those artifacts, but blurs the fine structures by producing cartoon-like features. The zoomed regions suggest that SGD-Net with $B=30$ can accurately reconstruct the fine details as good as its batch version using all the measurements.

Fig.~\ref{Fig:bsSpeed} provides several additional evaluations highlighting the influence of the number of steps and minibatch sizes. Fig.~\ref{Fig:bsSpeed}(a) shows the SNR performance versus number of unrolling steps. It can be observed that the average SNR values improve as we increase the number of model steps, and SGD-Net with $B = 60$ consistently achieves a similar performance with the full batch unrolled RED algorithm. Fig.~\ref{Fig:bsSpeed}(b) shows the performance of SGD-Net trained for $B \in \{30, 60, 90\}$ and U-RED when tested with different minibatch sizes. This results highlights the robustness of pre-trained SGD-Net to the changes in minibatch sizes. One can observe that U-RED degrades much faster than SGD-Net with $B=30$ when using smaller minibatches. Fig.~\ref{Fig:bsSpeed}(c) shows the computational time of using different minibatch sizes for reconstructing all the testing slices. We conducted 40 trials and 20 trails to calculate the average performance for each minibatch on GPU and CPU, respectively. 


\section{Conclusion}
\label{Sec:Conclusion}
The proposed SGD-Net method introduces stochastic approximation to the data-consistency layers of deep unfolding networks. Such approximations lead to complexity reductions during both training and testing of the complete network, making the network potentially applicable for problems with a large number of measurements. We provided extensive numerical results motivating the practical relevance of SGD-Net. Our results indicate that SGD-Net provides competitive imaging quality compared to the traditional and learning-based methods due to the training of the priors in conjunction with the forward model. In particular, SGD-Net using small minibatches achieves the SNR performance of the network using all the available measurements at a fraction of complexity. While our experiments focused on IDT and CT, the method is broadly applicable to many other imaging modalities such as optical diffraction tomography (ODT)~\cite{Kamilov.etal2016} and photoacoustic tomography (PAT)~\cite{Hauptmann.etal2018}, where the evaluation of the measurement model is computationally intensive. Additionally, while we implemented SGD-Net based on the RED framework, the idea can be used for other model-based deep learning architectures.

\appendix
\section{Proof of Theorem~\ref{Thm:MainConvRes}}
Assumption~\ref{As:LipschitzCon} leads to two bounds. By setting $\phibm_1 = \phibm_2=\phibm$ in Assumption~\ref{As:LipschitzCon}, we obtain the traditional Lipschitz continuity bound on the gradient
$$\|\nabla F(\thetabm_1; \phibm) -\nabla F(\thetabm_2; \phibm) \|_2^2 \leq L_{\theta}\|\thetabm_1 - \thetabm_2\|_2^2,$$
for all $\thetabm_1$, $\thetabm_2$, and $\phibm$, which directly leads to the traditional quadratic upper bound (see Lemma 1.2.3 in~\cite{Nesterov2004})
\begin{align}
\label{Eq:QuadBound}
F(\thetabm_1; \phibm) \leq F(\thetabm_2; \phibm) &+ \nabla F(\thetabm_2; \phibm)^{\Tsf}(\thetabm_1 - \thetabm_2) +\frac{ L_{\theta}}{2}\|\thetabm_1-\thetabm_2\|_2^2,
\end{align}
By setting $\thetabm_1=\thetabm_2=\thetabm$ in Assumption~\ref{As:LipschitzCon}, we obtain the following useful bound for our proof
\begin{equation}
\label{Eq:stcerrorBound}
\begin{aligned}
&\frac{1}{2}\|\nabla F(\thetabm; \phibm) -\nabla F(\thetabm; \phibmhat) \|_2^2 \leq \frac{ L_{\phi}^2}{2}\|\phibm - \phibmhat\|_2^2\\
&\Leftrightarrow \quad -\nabla F(\thetabm; \phibm)^{\Tsf}\nabla F(\thetabm; \phibmhat) + \frac{1}{2}\| F(\thetabm; \phibmhat)\|_2^2\\
&\quad\quad\quad\quad\quad\quad\quad\leq\frac{ L_{\phi}^2}{2}\|\phibm - \phibmhat\|_2^2 - \frac{1}{2}\|\nabla F(\thetabm; \phibm)\|_2^2\\
\end{aligned}
\end{equation}
The unbiasedness and boundedness of the variance of the stochastic gradient in Assumption~\ref{As:StochCost} implies that for any fixed vector $\thetabm$ and $\phibm$, we have that
\begin{equation}
\label{Eq:StochGradBound}
\begin{aligned}
&\E\left[\|\nablahat F(\thetabm;  \phibm)-  \nabla F(\thetabm; \phibm)\|_2^2 | \thetabm, \phibm\right] \leq \epsilon^2\\
&\Leftrightarrow \quad \E\left[\|\nablahat F(\thetabm; \phibm)\|_2^2 | \thetabm, \phibm\right] \leq\, \|\nabla F(\thetabm; \phibm)\|_2^2 + \epsilon^2,
\end{aligned}
\end{equation}
where the expectation is taken over the index $j \in \{1, \dots, M\}$ of the stochastic gradient.

Now, we are ready to establish the result in Theorem 1. Consider a single iteration of optimizing SGD-Net with SGD
$$\thetabm^{k+1}=\thetabm^k - \eta_k \nablahat F(\thetabm^k \phibmhat^k).$$
From the quadratic upper~\eqref{Eq:QuadBound}, we get
\begin{align}
\label{Eq:quadupperBound}
&  F(\thetabm^{k+1}; \phibm)  - F(\thetabm^{k}; \phibm)\\
\nonumber&\leq \nabla F(\thetabm^k; \phibm)^{\Tsf}(\thetabm^{k+1} - \thetabm^k) + \frac{ L_{\theta}}{2}\|\thetabm^{k+1}-\thetabm^k\|_2^2\\
\nonumber&= -\eta_{k}\nabla F(\thetabm^k; \phibm)^{\Tsf} \nablahat F(\thetabm^k; \phibmhat^k)+ \frac{\eta_{k}^2 L_{\theta}}{2}\|\nablahat F(\thetabm^k; \phibmhat^k)\|_2^2,
\end{align}
where $\phibm$ represents the true physical parameters and $\phibm^k$ the stochastic approximation.
By taking the conditional expectation with respect to the previous learning $\thetabm^k$ and physical paramters $\phihat^k$, using~\eqref{Eq:stcerrorBound} and~\eqref{Eq:StochGradBound}, we obtain
\begin{align}
\label{Eq:expectBound}
\nonumber& \E\left[F(\thetabm^{k+1}; \phibm) | \thetabm^k, \phihat^k\right] - F(\thetabm^k; \phibm)\\
\nonumber&\leq-\eta_{k}\nabla F(\thetabm^k; \phibm)^{\Tsf} \E\left[\nablahat F(\thetabm^k; \phibmhat^k) | \thetabm^k, \phibmhat^k\right] + \frac{\eta_{k}}{2}\|\nabla F(\thetabm^k; \phibmhat^k)\|_2^2 +\frac{\eta_{k}^2L_{\theta}\epsilon^2}{2}\\
\nonumber&\leq -\frac{\eta_{k}}{2}\|\nabla F(\thetabm^k; \phibm)\|_2^2+ \frac{ \eta_{k} L_{\phi}^2}{2}\|\phibm - \phibmhat^k\|_2^2 +\frac{\eta_{k}^2L_{\theta}\epsilon^2}{2},
\end{align}
where we also used the unbiasedness of the stochastic gradient $\nablahat F$ and the fact that $0 < \eta_k \leq 1/L_{\theta}$. By rearranging the terms, using Assumption~\ref{As:UnbiasedAssumption}, and taking the law of total expectation, we get
\begin{equation}
\begin{aligned}
\nonumber& \E\left[\eta_{k}\|\nabla F(\thetabm^k; \phibm)\|_2^2\right] \leq 2 (\E\left[F(\thetabm^k; \phibm)\right] - \E\left[F(\thetabm^{k+1}; \phibm)\right]) + \left(\frac{ \eta_{k} L_{\phi}^2\sigma^2}{B} + \eta_{k}^2 L_{\theta}\epsilon^2\right).
\end{aligned}
\end{equation}
By summing this bound over $0 \leq k \leq K-1$, we get the main result
\begin{equation}
\begin{aligned}
\nonumber&\sum_{k=0}^{K-1}\eta_{k} \E\left[ \|\nabla F(\thetabm^{k}; \phibm) \|_2^2\right]\\
&\leq{2(F(\thetabm^0; \phibm) - \E\left[F(\thetabm^K; \phibm))\right]+\sum_{k=0}^{K-1} (\frac{\eta_{k}L_{\phi}^2\sigma^2}{B} }+\eta_{k}^2 L_{\theta}\epsilon^2)\\
&\leq2(F(\thetabm^0; \phibm) - F(\thetabm^{\ast}; \phibm)) +\sum_{k=0}^{K-1} (\frac{\eta_{k}L_{\phi}^2\sigma^2}{B} +\eta_{k}^2 L_{\theta}\epsilon^2),
\end{aligned}
\end{equation}
where we used the fact that $F(\thetabm^{\ast}; \phibm) \leq \E\left[F(\thetabm^K; \phibm))\right]$, with $\thetabm^{\ast}$ denoting a global minimizer of $F$.

\section*{Acknowledgements}

Research presented in this article was supported by NSF award CCF-1813910 and the Laboratory Directed Research and Development program of Los Alamos National Laboratory under project number 20200061DR.

\bibliographystyle{IEEEtran}

\begin{thebibliography}{10}
\providecommand{\url}[1]{#1}
\csname url@samestyle\endcsname
\providecommand{\newblock}{\relax}
\providecommand{\bibinfo}[2]{#2}
\providecommand{\BIBentrySTDinterwordspacing}{\spaceskip=0pt\relax}
\providecommand{\BIBentryALTinterwordstretchfactor}{4}
\providecommand{\BIBentryALTinterwordspacing}{\spaceskip=\fontdimen2\font plus
\BIBentryALTinterwordstretchfactor\fontdimen3\font minus
  \fontdimen4\font\relax}
\providecommand{\BIBforeignlanguage}[2]{{%
\expandafter\ifx\csname l@#1\endcsname\relax
\typeout{** WARNING: IEEEtran.bst: No hyphenation pattern has been}%
\typeout{** loaded for the language `#1'. Using the pattern for}%
\typeout{** the default language instead.}%
\else
\language=\csname l@#1\endcsname
\fi
#2}}
\providecommand{\BIBdecl}{\relax}
\BIBdecl

\bibitem{Rudin.etal1992}
L.~I. Rudin, S.~Osher, and E.~Fatemi, ``Nonlinear total variation based noise
  removal algorithms,'' \emph{Physica D}, vol.~60, no. 1--4, pp. 259--268, Nov.
  1992.

\bibitem{Figueiredo.Nowak2001}
M.~A.~T. Figueiredo and R.~D. Nowak, ``Wavelet-based image estimation: An
  empirical {B}ayes approach using {J}effreys' noninformative prior,''
  \emph{IEEE Trans. Image Process.}, vol.~10, no.~9, pp. 1322--1331, Sep. 2001.

\bibitem{Elad.Aharon2006}
M.~Elad and M.~Aharon, ``Image denoising via sparse and redundant
  representations over learned dictionaries,'' \emph{IEEE Trans. Image
  Process.}, vol.~15, no.~12, pp. 3736--3745, Dec. 2006.

\bibitem{Danielyan.etal2012}
A.~Danielyan, V.~Katkovnik, and K.~Egiazarian, ``{BM3D} frames and variational
  image deblurring,'' \emph{IEEE Trans. Image Process.}, vol.~21, no.~4, pp.
  1715--1728, Apr. 2012.

\bibitem{Hu.etal2012}
Y.~Hu, S.~G. Lingala, and M.~Jacob, ``A fast majorize-minimize algorithm for
  the recovery of sparse and low-rank matrices,'' \emph{IEEE Trans. Image
  Process.}, vol.~21, no.~2, pp. 742--753, Feb. 2012.

\bibitem{Lefkimmiatis.etal2013}
S.~Lefkimmiatis, J.~P. Ward, and M.~Unser, ``Hessian {S}chatten-norm
  regularization for linear inverse problems,'' \emph{IEEE Trans. Image
  Process.}, vol.~22, no.~5, pp. 1873--1888, May 2013.

\bibitem{McCann.etal2017}
M.~T. McCann, K.~H. Jin, and M.~Unser, ``Convolutional neural networks for
  inverse problems in imaging: A review,'' \emph{IEEE Signal Process. Mag.},
  vol.~34, no.~6, pp. 85--95, 2017.

\bibitem{Lucas.etal2018}
A.~Lucas, M.~Iliadis, R.~Molina, and A.~K. Katsaggelos, ``Using deep neural
  networks for inverse problems in imaging: {B}eyond analytical methods,''
  \emph{IEEE Signal Process. Mag.}, vol.~35, no.~1, pp. 20--36, Jan. 2018.

\bibitem{Ongie.etal2020}
G.~Ongie, A.~Jalal, C.~A. Metzler, R.~G. Baraniuk, A.~G. Dimakis, and
  R.~Willett, ``Deep learning techniques for inverse problems in imaging,''
  \emph{IEEE J. Sel. Areas Inf. Theory}, vol.~1, no.~1, pp. 39--56, May 2020.

\bibitem{Ronneberger.etal2015}
O.~Ronneberger, P.~Fischer, and T.~Brox, ``{U}-{N}et: {C}onvolutional networks
  for biomedical image segmentation,'' in \emph{Proc. Med. Image. Comput.
  Comput. Assist. Intervent.}, 2015, pp. 234--241.

\bibitem{DJin.etal2017}
K.~H. {Jin}, M.~T. {McCann}, E.~{Froustey}, and M.~{Unser}, ``Deep
  convolutional neural network for inverse problems in imaging,'' \emph{IEEE
  Trans. Image Process.}, vol.~26, no.~9, pp. 4509--4522, Sep. 2017.

\bibitem{Kang.etal2017}
E.~Kang, J.~Min, and J.~C. Ye, ``A deep convolutional neural network using
  directional wavelets for low-dose x-ray {CT} reconstruction,'' \emph{Medical
  Physics}, vol.~44, no.~10, pp. e360--e375, 2017.

\bibitem{Chen.etal2017}
H.~Chen, Y.~Zhang, M.~K. Kalra, F.~Lin, Y.~Chen, P.~Liao, J.~Zhou, and G.~Wang,
  ``Low-dose {CT} with a residual encoder-decoder convolutional neural
  network,'' \emph{IEEE Trans. Med. Imag.}, vol.~36, no.~12, pp. 2524--2535,
  Dec. 2017.

\bibitem{Sun.etal2018}
Y.~Sun, Z.~Xia, and U.~S. Kamilov, ``Efficient and accurate inversion of
  multiple scattering with deep learning,'' \emph{Opt. Express}, vol.~26,
  no.~11, pp. 14\,678--14\,688, May 2018.

\bibitem{han.etal2018}
Y.~{Han} and J.~C. {Ye}, ``Framing {U-Net} via deep convolutional framelets:
  Application to sparse-view {CT},'' \emph{IEEE Trans. Med. Imag.}, vol.~37,
  no.~6, pp. 1418--1429, 2018.

\bibitem{Venkatakrishnan.etal2013}
S.~V. Venkatakrishnan, C.~A. Bouman, and B.~Wohlberg, ``Plug-and-play priors
  for model based reconstruction,'' in \emph{Proc. IEEE Global Conf. Signal
  Process. and Inf. Process.}, Austin, TX, USA, Dec. 3-5, 2013, pp. 945--948.

\bibitem{Romano.etal2017}
Y.~Romano, M.~Elad, and P.~Milanfar, ``The little engine that could:
  {R}egularization by denoising ({RED}),'' \emph{SIAM J. Imaging Sci.},
  vol.~10, no.~4, pp. 1804--1844, 2017.

\bibitem{Zhang.etal2017}
K.~Zhang, W.~Zuo, Y.~Chen, D.~Meng, and L.~Zhang, ``Beyond a {G}aussian
  denoiser: {R}esidual learning of deep {CNN} for image denoising,'' \emph{IEEE
  Trans. Image Process.}, vol.~26, no.~7, pp. 3142--3155, Jul. 2017.

\bibitem{Chan.etal2016}
S.~H. Chan, X.~Wang, and O.~A. Elgendy, ``Plug-and-play {ADMM} for image
  restoration: Fixed-point convergence and applications,'' \emph{IEEE Trans.
  Comp. Imag.}, vol.~3, no.~1, pp. 84--98, Mar. 2017.

\bibitem{Sreehari.etal2016}
S.~Sreehari, S.~V. Venkatakrishnan, B.~Wohlberg, G.~T. Buzzard, L.~F. Drummy,
  J.~P. Simmons, and C.~A. Bouman, ``Plug-and-play priors for bright field
  electron tomography and sparse interpolation,'' \emph{IEEE Trans. Comput.
  Imaging}, vol.~2, no.~4, pp. 408--423, Dec. 2016.

\bibitem{Kamilov.etal2017}
U.~S. Kamilov, H.~Mansour, and B.~Wohlberg, ``A plug-and-play priors approach
  for solving nonlinear imaging inverse problems,'' \emph{IEEE Signal. Proc.
  Let.}, vol.~24, no.~12, pp. 1872--1876, Dec. 2017.

\bibitem{Buzzard.etal2017}
G.~T. Buzzard, S.~H. Chan, S.~Sreehari, and C.~A. Bouman, ``Plug-and-play
  unplugged: {O}ptimization free reconstruction using consensus equilibrium,''
  \emph{SIAM J. Imaging Sci.}, vol.~11, no.~3, pp. 2001--2020, Sep. 2018.

\bibitem{Sun.etal2019a}
Y.~Sun, B.~Wohlberg, and U.~S. Kamilov, ``An online plug-and-play algorithm for
  regularized image reconstruction,'' \emph{IEEE Trans. Comput. Imaging},
  vol.~5, no.~3, pp. 395--408, Sep. 2019.

\bibitem{Reehorst.Schniter2019}
E.~T. Reehorst and P.~Schniter, ``Regularization by denoising: Clarifications
  and new interpretations,'' \emph{IEEE Trans. Comput. Imag.}, vol.~5, no.~1,
  pp. 52--67, Mar. 2019.

\bibitem{Ryu.etal2019}
E.~K. Ryu, J.~Liu, S.~Wang, X.~Chen, Z.~Wang, and W.~Yin, ``Plug-and-play
  methods provably converge with properly trained denoisers,'' in \emph{Proc.
  36th Int. Conf. Mach. Learn.}, vol.~97, Long Beach, CA, USA, Jun. 09--15
  2019, pp. 5546--5557.

\bibitem{Mataev.etal2019}
G.~Mataev, P.~Milanfar, and M.~Elad, ``Deep{RED}: Deep image prior powered by
  {RED},'' in \emph{Proc. {IEEE} Int. Conf. Comput. Vis. Workshops}, Oct. 2019,
  pp. 1--10.

\bibitem{Wu.etal2020}
Z.~{Wu}, Y.~{Sun}, A.~{Matlock}, J.~{Liu}, L.~{Tian}, and U.~S. {Kamilov},
  ``{SIMBA}: Scalable inversion in optical tomography using deep denoising
  priors,'' \emph{IEEE J. Sel. Topics Signal Process.}, pp. 1--1, 2020.

\bibitem{Liu.etal2020}
J.~{Liu}, Y.~{Sun}, C.~{Eldeniz}, W.~{Gan}, H.~{An}, and U.~S. {Kamilov},
  ``{RARE}: Image reconstruction using deep priors learned without ground
  truth,'' \emph{IEEE J. Sel. Topics Signal Process.}, pp. 1--1, 2020.

\bibitem{zhang2018ista}
J.~{Zhang} and B.~{Ghanem}, ``{ISTA-Net}: {I}nterpretable optimization-inspired
  deep network for image compressive sensing,'' in \emph{Proc. {IEEE} Conf.
  Comput. Vision Pattern Recognit.}, 2018, pp. 1828--1837.

\bibitem{Yang.etal2016}
Y.~Yang, J.~Sun, H.~Li, and Z.~Xu, ``Deep {ADMM}-{N}et for compressive sensing
  {MRI},'' in \emph{Proc. Advances Neural Inf. Process. Syst.}, 2016, pp.
  10--18.

\bibitem{Hauptmann.etal2018}
A.~{Hauptmann}, F.~{Lucka}, M.~{Betcke}, N.~{Huynh}, J.~{Adler}, B.~{Cox},
  P.~{Beard}, S.~{Ourselin}, and S.~{Arridge}, ``Model-based learning for
  accelerated, limited-view 3-d photoacoustic tomography,'' \emph{IEEE Trans.
  Med. Imag.}, vol.~37, no.~6, pp. 1382--1393, 2018.

\bibitem{Adler.etal2018}
J.~{Adler} and O.~{{\"O}ktem}, ``Learned primal-dual reconstruction,''
  \emph{IEEE Trans. Med. Imag.}, vol.~37, no.~6, pp. 1322--1332, June 2018.

\bibitem{Aggarwal.etal2019}
H.~K. {Aggarwal}, M.~P. {Mani}, and M.~{Jacob}, ``Modl: Model-based deep
  learning architecture for inverse problems,'' \emph{IEEE Trans. Med. Imag.},
  vol.~38, no.~2, pp. 394--405, Feb. 2019.

\bibitem{Hosseini.etal2019}
S.~A. Hosseini, B.~Yaman, S.~Moeller, M.~Hong, and M.~Akcakaya, ``Dense
  recurrent neural networks for accelerated {MRI}: {H}istory-cognizant
  unrolling of optimization algorithms,'' \emph{IEEE J. Sel. Topics Signal
  Process.}, vol.~14, no.~6, pp. 1280--1291, Oct. 2020.

\bibitem{Chun.etal2020}
I.~Y. {Chun}, Z.~{Huang}, H.~{Lim}, and J.~{Fessler}, ``{Momentum-Net}: {F}ast
  and convergent iterative neural network for inverse problems,'' \emph{IEEE
  Trans. Patt. Anal. and Machine Intell.}, pp. 1--1, 2020.

\bibitem{Yaman.etal2020}
B.~Yaman, S.~A.~H. Hosseini, S.~Moeller, J.~Ellermann, K.~U{\u g}urbil, and
  M.~Ak{\c c}akaya, ``Self-supervised learning of physics-guided reconstruction
  neural networks without fully sampled reference data,'' \emph{Magn. Reson.
  Med.}, Jul. 2020.

\bibitem{Aggarwal.etal2020}
H.~K. {Aggarwal} and M.~{Jacob}, ``{J-MoDL}: {J}oint model-based deep learning
  for optimized sampling and reconstruction,'' \emph{IEEE J. Sel. Topics Signal
  Process.}, vol.~14, no.~6, pp. 1151--1162, 2020.

\bibitem{Kellman.etal2020}
M.~{Kellman}, K.~{Zhang}, E.~{Markley}, J.~{Tamir}, E.~{Bostan}, M.~{Lustig},
  and L.~{Waller}, ``Memory-efficient learning for large-scale computational
  imaging,'' \emph{IEEE Trans. Comp. Imag.}, vol.~6, pp. 1403--1414, 2020.

\bibitem{Bottou.Bousquet2007}
L.~Bottou and O.~Bousquet, ``The tradeoffs of large scale learning,'' in
  \emph{Proc. Advances Neural Inf. Process. Syst.}, Vancouver, BC, Canada, Dec.
  3-6, 2007, pp. 161--168.

\bibitem{Bertsekas2011}
D.~P. Bertsekas, ``Incremental proximal methods for large scale convex
  optimization,'' \emph{Math. Program. Ser. B}, vol. 129, pp. 163--195, 2011.

\bibitem{Kim.etal2013}
D.~{Kim}, D.~{Pal}, J.~{Thibault}, and J.~A. {Fessler}, ``Accelerating ordered
  subsets image reconstruction for {X}-ray {CT} using spatially nonuniform
  optimization transfer,'' \emph{IEEE Trans. Med. Imag.}, vol.~32, no.~11, pp.
  1965--1978, Nov. 2013.

\bibitem{Bottou.etal2018}
L.~Bottou, F.~E. Curtis, and J.~Nocedal, ``Optimization methods for large-scale
  machine learning,'' \emph{SIAM Rev.}, vol.~60, no.~2, pp. 223--311, 2018.

\bibitem{Ling.etal18}
R.~Ling, W.~Tahir, H.-Y. Lin, H.~Lee, and L.~Tian, ``High-throughput intensity
  diffraction tomography with a computational microscope,'' \emph{Biomed. Opt.
  Express}, vol.~9, no.~5, pp. 2130--2141, May 2018.

\bibitem{Kak.Slaney1988}
A.~C. Kak and M.~Slaney, \emph{Principles of Computerized Tomographic
  Imaging}.\hskip 1em plus 0.5em minus 0.4em\relax {IEEE}, 1988.

\bibitem{Parikh.Boyd2014}
N.~Parikh and S.~Boyd, ``Proximal algorithms,'' \emph{Foundations and Trends in
  Optimization}, vol.~1, no.~3, pp. 123--231, 2014.

\bibitem{Eckstein.Bertsekas1992}
J.~Eckstein and D.~P. Bertsekas, ``On the {D}ouglas-{R}achford splitting method
  and the proximal point algorithm for maximal monotone operators,''
  \emph{Mathematical Programming}, vol.~55, pp. 293--318, 1992.

\bibitem{Figueiredo.Nowak2003}
M.~A.~T. Figueiredo and R.~D. Nowak, ``An {EM} algorithm for wavelet-based
  image restoration,'' \emph{IEEE Trans. Image Process.}, vol.~12, no.~8, pp.
  906--916, Aug. 2003.

\bibitem{Bect.etal2004}
J.~Bect, L.~Blanc-Feraud, G.~Aubert, and A.~Chambolle, ``A $\ell_1$-unified
  variational framework for image restoration,'' in \emph{Proc. Euro. Conf.
  Comp. Vis.}, vol. 3024, New York, 2004, pp. 1--13.

\bibitem{Daubechies.etal2004}
I.~Daubechies, M.~Defrise, and C.~D. Mol, ``An iterative thresholding algorithm
  for linear inverse problems with a sparsity constraint,'' \emph{Commun. Pure
  Appl. Math.}, vol.~57, no.~11, pp. 1413--1457, Nov. 2004.

\bibitem{Beck.Teboulle2009b}
A.~Beck and M.~Teboulle, \emph{Convex Optimization in Signal Processing and
  Communications}.\hskip 1em plus 0.5em minus 0.4em\relax Cambridge, 2009, ch.
  Gradient-Based Algorithms with Applications to Signal Recovery Problems, pp.
  42--88.

\bibitem{Afonso.etal2010}
M.~V. Afonso, J.~M.Bioucas-Dias, and M.~A.~T. Figueiredo, ``Fast image recovery
  using variable splitting and constrained optimization,'' \emph{IEEE Trans.
  Image Process.}, vol.~19, no.~9, pp. 2345--2356, Sep. 2010.

\bibitem{Ng.etal2010}
M.~K. Ng, P.~Weiss, and X.~Yuan, ``Solving constrained total-variation image
  restoration and reconstruction problems via alternating direction methods,''
  \emph{SIAM J. Sci. Comput.}, vol.~32, no.~5, pp. 2710--2736, Aug. 2010.

\bibitem{Boyd.etal2011}
S.~Boyd, N.~Parikh, E.~Chu, B.~Peleato, and J.~Eckstein, ``Distributed
  optimization and statistical learning via the alternating direction method of
  multipliers,'' \emph{Foundations and Trends in Machine Learning}, vol.~3,
  no.~1, pp. 1--122, July 2011.

\bibitem{Zhao.etal2017}
H.~{Zhao}, O.~{Gallo}, I.~{Frosio}, and J.~{Kautz}, ``Loss functions for image
  restoration with neural networks,'' \emph{IEEE Trans. Comput. Imaging},
  vol.~3, no.~1, pp. 47--57, Mar. 2017.

\bibitem{Gregor.LeCun2010}
K.~Gregor and Y.~LeCun, ``Learning fast approximation of sparse coding,'' in
  \emph{Proc. 27th Int. Conf. Mach. Learn.}, Haifa, Israel, Jun. 21-24, 2010,
  pp. 399--406.

\bibitem{Schmidt.Roth2014}
U.~Schmidt and S.~Roth, ``Shrinkage fields for effective image restoration,''
  in \emph{Proc. {IEEE} Conf. Comput. Vis. Pattern Recognit.}, Columbus, OH,
  USA, Jun. 23-28, 2014, pp. 2774--2781.

\bibitem{Chen.etal2015}
Y.~Chen, W.~Yu, and T.~Pock, ``On learning optimized reaction diffuction
  processes for effective image restoration,'' in \emph{Proc. {IEEE} Conf.
  Comput. Vis. Pattern Recognit.}, Boston, MA, USA, Jun. 8-10, 2015, pp.
  5261--5269.

\bibitem{Kamilov.Mansour2016}
U.~S. Kamilov and H.~Mansour, ``Learning optimal nonlinearities for iterative
  thresholding algorithms,'' \emph{IEEE Signal Process. Lett.}, vol.~23, no.~5,
  pp. 747--751, May 2016.

\bibitem{Bostan.etal2018}
E.~Bostan, U.~S. Kamilov, and L.~Waller, ``Learning-based image reconstruction
  via parallel proximal algorithm,'' \emph{IEEE Signal Process. Lett.},
  vol.~25, no.~7, pp. 989--993, Jul. 2018.

\bibitem{Schlemper.etal2018}
J.~Schlemper, J.~Caballero, J.~V. Hajnal, A.~N. Price, and D.~Rueckert, ``A
  deep cascade of convolutional neural networks for dynamic {MR} image
  reconstruction,'' \emph{IEEE Trans. Med. Imag.}, vol.~37, no.~2, pp.
  491--503, Feb. 2018.

\bibitem{Biswas.etal2019}
S.~Biswas, H.~K. Aggarwal, and M.~Jacob, ``Dynamic {MRI} using model‐based
  deep learning and {SToRM} priors: {MoDL}‐{SToRM},'' \emph{Magn. Reson.
  Med.}, vol.~82, no.~1, pp. 485--494, Jul. 2019.

\bibitem{Dabov.etal2007}
K.~Dabov, A.~Foi, V.~Katkovnik, and K.~Egiazarian, ``Image denoising by sparse
  {3-D} transform-domain collaborative filtering,'' \emph{IEEE Trans. Image
  Process.}, vol.~16, no.~16, pp. 2080--2095, Aug. 2007.

\bibitem{Brifman.etal2016}
A.~Brifman, Y.~Romano, and M.~Elad, ``Turning a denoiser into a super-resolver
  using plug and play priors,'' in \emph{Proc. {IEEE} Int. Conf. Image Proc.},
  Phoenix, AZ, USA, Sep. 25-28, 2016, pp. 1404--1408.

\bibitem{Ahmad.etal2019}
R.~Ahmad, C.~A. Bouman, G.~T. Buzzard, S.~H. Chan, S.Liu, E.~Reehorst, and
  P.~Schniter, ``Plug-and-play methods for magnetic resonance imaging: Using
  denoisers for image recovery,'' \emph{IEEE Signal Process. Mag.}, vol.~37,
  no.~1, pp. 105--116, Jan. 2020.

\bibitem{Metzler.etal2018}
C.~Metzler, P.~Schniter, A.~Veeraraghavan, and R.~Baraniuk, ``pr{D}eep: Robust
  phase retrieval with a flexible deep network,'' in \emph{Proc. 36th Int.
  Conf. Mach. Learn.}, Stockholmsm{\"a}ssan, Stockholm Sweden, Jul. 10--15
  2018, pp. 3501--3510.

\bibitem{Sun.etal2019c}
Y.~Sun, J.~Liu, and U.~S. Kamilov, ``Block coordinate regularization by
  denoising,'' in \emph{Proc. Advances Neural Inf. Process. Syst.}, Vancouver,
  BC, Canada, Dec. 8-14, 2019, pp. 380--390.

\bibitem{zhang.etal2020a}
K.~Zhang, L.~V. Gool, and R.~Timofte, ``Deep unfolding network for image
  super-resolution,'' in \emph{Proc. {IEEE} Conf. Comput. Vis. Pattern
  Recognit.}, Jun. 2020, pp. 3217--3226.

\bibitem{Wu.etal2018}
Y.~Wu and K.~He, ``Group normalization,'' in \emph{Proc. Euro. Conf. Comp.
  Vis.}, Sep. 2018, pp. 3--19.

\bibitem{Robbins.Monro1951}
H.~Robbins and S.~Monro, ``A stochastic approximation method,'' \emph{The
  Annals of Mathematical Statistics}, vol.~22, no.~3, pp. 400--407, September
  1951.

\bibitem{Nesterov2004}
Y.~Nesterov, \emph{Introductory Lectures on Convex Optimization: A Basic
  Course}.\hskip 1em plus 0.5em minus 0.4em\relax Kluwer Academic Publishers,
  2004.

\bibitem{Ghadimi.Lan2016}
S.~Ghadimi and G.~Lan, ``Accelerated gradient methods for nonconvex nonlinear
  and stochastic programming,'' \emph{Math. Program. Ser. A}, vol. 156, no.~1,
  pp. 59--99, Mar. 2016.

\bibitem{Wu.etal2019}
Z.~Wu, Y.~Sun, J.~Liu, and U.~S. Kamilov, ``Online regularization by denoising
  with applications to phase retrieval,'' in \emph{Proc. {IEEE} Int. Conf.
  Comput. Vis. Workshops}, Oct. 2019, pp. 1--9.

\bibitem{Beck.Teboulle2009a}
A.~Beck and M.~Teboulle, ``Fast gradient-based algorithm for constrained total
  variation image denoising and deblurring problems,'' \emph{IEEE Trans. Image
  Process.}, vol.~18, no.~11, pp. 2419--2434, November 2009.

\bibitem{Wang.etal2004}
Z.~Wang, A.~C. Bovik, H.~R. Sheikh, and E.~P. Simoncelli, ``Image quality
  assessment: from error visibility to structural similarity,'' \emph{IEEE
  Trans. Image Process.}, vol.~13, no.~4, pp. 600--612, Apr 2004.

\bibitem{knoll2020fastmri}
{F. {Knoll} \emph{et al.}}, ``{fastMRI}: A publicly available raw k-space and
  {DICOM} dataset of knee images for accelerated {MR} image reconstruction
  using machine learning,'' \emph{Radiology: Artificial Intelligence}, vol.~2,
  no.~1, p. e190007, 2020.

\bibitem{Smith2017}
L.~N. {Smith}, ``Cyclical learning rates for training neural networks,'' in
  \emph{2017 IEEE Winter Conference on Applications of Computer Vision}, Mar.
  2017, pp. 464--472.

\bibitem{mccollough2016tu}
C.~McCollough, ``{TU-FG-207A-04}: Overview of the low dose {CT} grand
  challenge,'' \emph{Med. Phys}, vol.~43, no. 6Part35, pp. 3759--3760, 2016.

\bibitem{adler.etal2017}
J.~Adler and O.~{\"O}ktem, ``Solving ill-posed inverse problems using iterative
  deep neural networks,'' \emph{Inverse Problems}, vol.~33, no.~12, p. 124007,
  2017.

\bibitem{vanAarle.etal16}
{W. van Aarle \emph{et al.}}, ``Fast and flexible x-ray tomography using the
  {ASTRA} toolbox,'' \emph{Opt. Express}, vol.~24, no.~22, pp.
  25\,129--25\,147, Oct. 2016.

\bibitem{Gupta2018}
H.~{Gupta}, K.~H. {Jin}, H.~Q. {Nguyen}, M.~T. {McCann}, and M.~{Unser},
  ``C{NN}-based projected gradient descent for consistent ct image
  reconstruction,'' \emph{IEEE Trans. Med. Imag.}, vol.~37, no.~6, pp.
  1440--1453, Jun. 2018.

\bibitem{Kingma.Ba2015}
D.~Kingma and J.~Ba, ``Adam: {A} method for stochastic optimization,'' in
  \emph{Proc. Int. Conf. on Learn. Represent.}, 2015.

\bibitem{Zhang2020Why}
J.~{Zhang}, T.~{He}, S.~{Sra}, and A.~{Jadbabaie}, ``Why gradient clipping
  accelerates training: A theoretical justification for adaptivity,'' in
  \emph{Proc. Int. Conf. on Learn. Represent.}, 2020.

\bibitem{Kamilov.etal2016}
U.~S. Kamilov, I.~N. Papadopoulos, M.~H. Shoreh, A.~Goy, C.~Vonesch, M.~Unser,
  and D.~Psaltis, ``Optical tomographic image reconstruction based on beam
  propagation and sparse regularization,'' \emph{IEEE Trans. Comp. Imag.},
  vol.~2, no.~1, pp. 59--70, Mar. 2016.

\end{thebibliography}


\end{document}